\documentclass[12pt]{article}
\usepackage[english]{babel}

\usepackage{fullpage}
\usepackage{cite}

\usepackage{multicol}
\newcommand{\cH}{\mathcal{H}}
\newcommand{\cL}{\mathcal{L}}
\newcommand*{\dI}{\mathbb{I}}
\newcommand*{\cD}{\mathcal{D}}

\newcommand*{\cN}{\pazocal{N}}
\newcommand*{\cM}{\mathcal{M}}

\newcommand*{\eps}{\varepsilon}
\newcommand*{\Haar}{{\rm Haar}}
\newcommand*{\unif}{{\rm Uniform}}
\newcommand*{\ptr}[2]{\mathrm{Tr}_{#1}\left[#2\right]}
\newcommand*{\proj}[1]{\ketbra{#1}}
\newcommand*{\prob}[1]{\mathbb{P}\left(#1 \right)}
\newcommand*{\ex}[1]{\mathbb{E}\left[#1 \right]}
\newcommand{\C}{\mathbb{C}}

\usepackage{newpxtext,newpxmath}

\let\coloneqq\relax

\usepackage[utf8]{inputenc}
\usepackage{amsthm}
\usepackage{amssymb}
\usepackage{amsmath}
\usepackage{bbold}
\usepackage{bbm}
\usepackage[pdftex, backref=page]{hyperref}
\usepackage{braket}
\usepackage{dsfont}
\usepackage{mathdots}
\usepackage{mathtools}
\usepackage{enumerate}
\usepackage[shortlabels]{enumitem}
\usepackage{csquotes}
\usepackage{stmaryrd}
\usepackage[cal=boondox]{mathalfa}
\usepackage{graphicx}
\usepackage{stackengine}
\usepackage{scalerel}
\usepackage{tensor}       
\usepackage{array}
\usepackage{makecell}
\newcolumntype{x}[1]{>{\centering\arraybackslash}p{#1}}
\usepackage{tikz}
\usepackage{pgfplots}
\usetikzlibrary{shapes.geometric, shapes.misc, positioning, arrows, arrows.meta, decorations.pathreplacing, decorations.pathmorphing, patterns, angles, quotes, calc}
\usepackage{booktabs}
\usepackage{xfrac}
\usepackage{siunitx}
\usepackage{centernot}
\usepackage{comment}
\usepackage{chngcntr}
\usepackage{caption}
\usepackage{subcaption}

\newtheorem{thm}{Theorem}
\newtheorem*{thm*}{Theorem}

\newtheorem*{prop*}{Proposition}
\newtheorem{lemma}[thm]{Lemma}
\newtheorem*{lemma*}{Lemma}

\newtheorem*{cor*}{Corollary}

\newtheorem*{cj*}{Conjecture}

\newtheorem*{Def*}{Definition}

\newtheorem*{question*}{Question}

\newtheorem*{problem*}{Problem}

\makeatletter
\def\thmhead@plain#1#2#3{%
  \thmname{#1}\thmnumber{\@ifnotempty{#1}{ }\@upn{#2}}%
  \thmnote{ {\the\thm@notefont#3}}}
\let\thmhead\thmhead@plain
\makeatother

\theoremstyle{definition}

\newcommand{\bb}{\begin{equation}\begin{aligned}\hspace{0pt}}
\newcommand{\bbb}{\begin{equation*}\begin{aligned}}
\newcommand{\ee}{\end{aligned}\end{equation}}
\newcommand{\eee}{\end{aligned}\end{equation*}}
\newcommand*{\coloneqq}{\mathrel{\vcenter{\baselineskip0.5ex \lineskiplimit0pt \hbox{\scriptsize.}\hbox{\scriptsize.}}} =}

\newcommand{\eqt}[1]{\stackrel{\mathclap{\scriptsize \mbox{#1}}}{=}}
\newcommand{\leqt}[1]{\stackrel{\mathclap{\scriptsize \mbox{#1}}}{\leq}}

\newcommand{\geqt}[1]{\stackrel{\mathclap{\scriptsize \mbox{#1}}}{\geq}}
\newcommand{\ketbra}[1]{\ket{#1}\!\!\bra{#1}}

\renewcommand{\epsilon}{\varepsilon}

\newcommand{\id}{\mathds{1}}

\DeclareMathOperator{\Tr}{Tr}

\DeclareMathAlphabet{\pazocal}{OMS}{zplm}{m}{n}

\newcommand{\EE}{\pazocal{E}}

\newcommand{\lsmatrix}{\left(\begin{smallmatrix}}
\newcommand{\rsmatrix}{\end{smallmatrix}\right)}

\stackMath

\stackMath

\makeatletter
\newcommand*\rel@kern[1]{\kern#1\dimexpr\macc@kerna}
\newcommand*\widebar[1]{%
  \begingroup
  \def\mathaccent##1##2{%
    \rel@kern{0.8}%
    \overline{\rel@kern{-0.8}\macc@nucleus\rel@kern{0.2}}%
    \rel@kern{-0.2}%
  }%
  \macc@depth\@ne
  \let\math@bgroup\@empty \let\math@egroup\macc@set@skewchar
  \mathsurround\z@ \frozen@everymath{\mathgroup\macc@group\relax}%
  \macc@set@skewchar\relax
  \let\mathaccentV\macc@nested@a
  \macc@nested@a\relax111{#1}%
  \endgroup
}

\counterwithin*{equation}{part}
\counterwithin*{thm}{part}
\counterwithin*{figure}{part}

\tikzset{meter/.append style={draw, inner sep=10, rectangle, font=\vphantom{A}, minimum width=30, line width=.8, path picture={\draw[black] ([shift={(.1,.3)}]path picture bounding box.south west) to[bend left=50] ([shift={(-.1,.3)}]path picture bounding box.south east);\draw[black,-latex] ([shift={(0,.1)}]path picture bounding box.south) -- ([shift={(.3,-.1)}]path picture bounding box.north);}}}
\tikzset{roundnode/.append style={circle, draw=black, fill=gray!20, thick, minimum size=10mm}}
\tikzset{squarenode/.style={rectangle, draw=black, fill=none, thick, minimum size=10mm}}

\definecolor{Blues5seq1}{RGB}{239,243,255}
\definecolor{Blues5seq2}{RGB}{189,215,231}
\definecolor{Blues5seq3}{RGB}{107,174,214}
\definecolor{Blues5seq4}{RGB}{49,130,189}
\definecolor{Blues5seq5}{RGB}{8,81,156}

\definecolor{Greens5seq1}{RGB}{237,248,233}
\definecolor{Greens5seq2}{RGB}{186,228,179}
\definecolor{Greens5seq3}{RGB}{116,196,118}
\definecolor{Greens5seq4}{RGB}{49,163,84}
\definecolor{Greens5seq5}{RGB}{0,109,44}

\definecolor{Reds5seq1}{RGB}{254,229,217}
\definecolor{Reds5seq2}{RGB}{252,174,145}
\definecolor{Reds5seq3}{RGB}{251,106,74}
\definecolor{Reds5seq4}{RGB}{222,45,38}
\definecolor{Reds5seq5}{RGB}{165,15,21}

\allowdisplaybreaks

\usepackage[most,breakable]{tcolorbox}
\newenvironment{boxedthm}[1]%
	{\expandafter\ifstrequal\expandafter{#1}{orange}{\begin{tcolorbox}[colback=red!15,colframe=orange!15,breakable,enhanced]}{\begin{tcolorbox}[colback=Blues5seq1,colframe=Blues5seq5,breakable,enhanced]}}%
	{\end{tcolorbox}}

\usepackage{authblk}

\renewcommand{\EE}[1]{\underset{\scaleobj{.8}{#1}}{\mathds{E}\,}}

\allowdisplaybreaks[1] 
\usepackage{graphicx,import,xcolor,transparent} 
\usepackage{enumitem}

\begin{document}

\author[1]{Aadil Oufkir\thanks{\texttt{aadil.oufkir@gmail.com}}}
\affil[1]{{Mohammed VI Polytechnic University, Rocade Rabat-Salé, Technopolis, Morocco}}
\author[2]{Filippo Girardi\thanks{\texttt{filippo.girardi@sns.it}}}
\affil[2]{Scuola Normale Superiore, Piazza dei Cavalieri 7, 56126 Pisa, Italy}

\title{Improved Lower Bounds for Learning Quantum Channels in Diamond Distance}

\date{}
\setcounter{Maxaffil}{0}
\renewcommand\Affilfont{\itshape\small}

\maketitle\vspace{-4ex}
\begin{abstract} 

We prove that learning an unknown quantum channel with input dimension $d_A$, output dimension $d_B$, and Choi rank $r$ to diamond distance $\varepsilon$ requires $
\Omega\!\left( \frac{d_A d_B r}{\varepsilon \log(d_B r / \varepsilon)} \right)$ channel 
queries when $d_A= rd_B$, and $
\Omega\!\left( \frac{d_A d_B r}{\varepsilon^2 \log(d_B r / \varepsilon)} \right)$ channel queries when  $d_A\le rd_B/2$. These lower bounds improve upon the best previous $\Omega(d_A d_B r)$ bound by introducing explicit, near-optimal $\varepsilon$-dependence. Moreover,  when $d_A\le rd_B/2$,   the lower bound is optimal up to a logarithmic factor. The proof constructs  ensembles of channels that are well separated in diamond norm yet admit Stinespring isometries that are close in operator norm.
\end{abstract}

\section{Introduction}\label{sec:intro}
In~\cite{AMele2025,chen2025quantumchanneltomographyestimation} it is proved that there exists a quantum learning algorithm that uses
    \bb
        N = O \left(\frac{d_Ad_Br}{\epsilon^2}\right)
    \ee
    parallel queries of any unknown channel $\cN$ with input dimension $d_A$, output dimension $d_B$, and Choi rank $r$ and, with probability at least $2/3$,  outputs a classical description of a channel $\hat{\cN}$ which is distant at most $\epsilon$ from $\cN$ in diamond distance. Moreover, in~\cite{Girardi2025Dec} it is proved that any quantum algorithm learning $\cN$ up to constant error with success probability at least $2/3$ needs
    \begin{equation}
        N= \Omega(d_Ad_Br)
    \end{equation}
    queries of $\cN$ at least. The aim of our work is to improve the lower bound in order to make it dependent on the diamond distance $\epsilon$. More precisely, our main result (see Theorem \ref{thm:main}) identifies the new lower bounds
    
    \bb\label{eq:new}
            N = \begin{cases}
                \Omega\left(\frac{d_Ad_Br}{\eps \log(d_Br/\eps)} \right) &\text{ when } \quad d_A = rd_B, \\ 
                 \Omega\left(\frac{d_Ad_Br}{\eps^2 \log(d_Br/\eps)} \right) & \text{ when }\quad d_A\le rd_B/2.
            \end{cases}
        \ee
    This result, combined   with the upper bound of \cite{chen2025quantumchanneltomographyestimation},  shows that  the optimal dependency in the precision parameter $\eps$ is $\widetilde{\Theta}(\frac{1}{\eps})$ when $d_A = d_Br$. 
     In particular, for learning unitary channels, our result recovers the optimal  bound $\Theta(\frac{d^2}{\eps})$ of \cite{haah2023query} up to a logarithmic factor, via a different proof strategy that applies specifically in the coherent setting. Moreover, our approach generalises to non-unitary channels. 
    Furthermore, in the regime  $d_A\le rd_B/2$, our lower bound is  optimal, up to a logarithmic factor, given the upper bounds of \cite{AMele2025,chen2025quantumchanneltomographyestimation}.

    The main idea to prove our lower bounds consists of two steps. The first one is the proof of a general lower bound (Theorem \ref{thm:gen-LB}), which leverage any arbitrary ensemble of channels $\{\Phi_i\}$ that are pairwise far in diamond distance, yet whose Stinespring isometries are pairwise close in operator norm. The resulting lower bound then scales logarithmically with the size of the ensemble and inversely with the distance of the Stinespring isometries. 
    The second step is the actual construction of  suitable ensembles of channels to prove our lower bounds. 
    A natural strategy to this end could be the use of existing packing nets. However, as we show in Appendix \ref{sec:LB-sqrt(eps)}, this approach -- although simpler -- yields a weaker bound:
    \begin{align}
        N \ge  \Omega\left(\frac{d_Ad_Br}{\sqrt{\eps}\log(d_Br/\sqrt{\eps})}\right). 
    \end{align}
    Instead, using a probabilistic approach, we construct two  families of random Stinespring isometries which are sufficiently close in operator norm, but with positive probability engender channels which are far enough in diamond distance. Such argument ensures the existence of two ensembles that produce the desired lower bounds (Theorem \ref{thm:main}).\\
    
    The remainder of the manuscript is organised as follows. In Section~\ref{sec:intro} we introduce the notation and the definitions that we are going to use in the paper. In Section~\ref{sec:LB} we state and prove Theorem \ref{thm:gen-LB}, i.e. the general lower bound on channel learning constructed in terms of ensembles of quantum channels. In Section~\ref{sec:improved} we prove the lower bound \eqref{eq:new} (see Theorem \ref{thm:main}) leveraging two families of particular isometries in order to construct  suitable ensembles of channels to be used in Theorem~\ref{thm:gen-LB}.  In the Appendix we provide the proofs that were deferred in the previous sections to improve readability.

\subsection{ Notation}
All the Hilbert spaces that we are going to consider are supposed to be finite-dimensional. Let $\cH_A \cong \C^{d_A}$ and $\cH_B \cong \C^{d_B}$ denote input and output spaces. We write $\cL(\cH)$ for linear operators on $\cH$, and $\cD(\cH)$ for quantum states (positive semi-definite  operators with unit trace). The operator norm is $\|X\|_\mathrm{op} = \sup_{\|\psi\|=1} \| X\ket{\psi} \|$, and the trace norm is $\|X\|_1 = \Tr\sqrt{X^\dagger X}$. For $\rho \in \cD(\cH)$, the von Neumann entropy is $S(\rho) = -\Tr[\rho \log \rho]$ and the conditional entropy of a bipartite state $\rho_{AB}$ is $S(A|B)_{\rho} = S(AB)_{\rho} - S(B)_{\rho}$. 
All logarithms are in base $\mathrm{e}$.

\subsection{Quantum channels and their representations}

A \emph{quantum channel} $\Phi: \cL(\cH_A) \to \cL(\cH_B)$ is a completely positive trace-preserving (CPTP) map. It can be written in the Kraus representation as 
$\Phi(X) = \sum_{i=1}^r K_i X K_i^\dagger$ with $\sum_i K_i^\dagger K_i = \id_A$. The minimal $r$ is called  the \emph{Kraus rank}. 

For any channel $\Phi$, there exists an \emph{isometry} $V: \cH_A \to \cH_B \otimes \cH_E$ ($V^\dagger V = \id_A$) such that $\Phi(X) = \Tr_E(V X V^\dagger)$. Such  isometry $V$ is called \emph{Stinespring dilation} of $\Phi$. The minimal  dimension of $ \cH_E$ equals the Kraus rank. 

The \emph{Choi state} of the channel $\Phi$ is 
\bb
J_\Phi \coloneqq ({\rm Id}_{A'} \otimes \Phi  )(\ket{\Psi}\bra{\Psi}) \in \cL(\cH_{A'} \otimes \cH_B),
\ee
where $\ket{\Psi}_{A'A} \coloneqq \frac{1}{\sqrt{d_A}}\sum_{i=1}^{d_A} \ket{i}_{A'} \otimes \ket{i}_{A}$ is the normalised maximally entangled state between $A'$ and $A$. The linear map $\Phi$ is CPTP if and only if $J_\Phi \geq 0$ and $\Tr_B J_\Phi = \id_{A'}/d_A$.

The \emph{Choi rank} $\operatorname{rank}_{\text{Choi}}{\Phi} \coloneqq \operatorname{rank}(J_\Phi)$ equals the  Kraus rank and the minimal environment dimension of Stinespring dilations.

\subsection{Channel ensembles with distance constraints}

Let $\mathcal{C}(d_A,d_B,r)$ denote the set of quantum channels 
$\Phi:\cL(\cH_A)\to\cL(\cH_B)$ with Choi rank at most $r$, i.e.
\bb
    \mathcal{C}(d_A,d_B,r) \coloneqq \{\Phi \text{ quantum channel} \mid 
    \operatorname{rank}_{\text{Choi}}{\Phi} \leq r\}.
\ee

For a channel $\Phi \in \mathcal{C}(d_A,d_B,r)$, let $V_\Phi:\cH_A \to \cH_B \otimes \cH_E$ be a Stinespring isometry with minimal environment dimension $d_E \le  r$.

For channels $\Phi,\Psi: \cL(\cH_A) \to \cL(\cH_B)$, the diamond distance is defined as 
\bb
\|\Phi - \Psi\|_\diamond \coloneqq \sup_{\rho_{RA} \in \cD(\cH_R \otimes \cH_A)} \|{({\rm Id}_R\otimes \Phi )(\rho_{RA}) - ({\rm Id}_R\otimes\Psi )(\rho_{RA})}\|_1,
\ee
where the supremum is over all auxiliary spaces $\cH_R$ and states $\rho_{RA}$.

We say that two channels $\Phi,\Psi \in \mathcal{C}(d_A,d_B,r)$ are 
$\varepsilon$-\emph{diamond far} if
\begin{equation}
    \|\Phi - \Psi\|_{\diamond} > \varepsilon.
\end{equation}
We say that their Stinespring isometries are $\eta$-\emph{operator norm close} if there exist choices of isometries $V_\Phi, V_\Psi$ such that
\begin{equation}
    \|V_\Phi - V_\Psi\|_{\mathrm{op}} \leq \eta.
\end{equation}

Finally, we define $\mathcal{E}(d_A,d_B,r,\varepsilon,\eta)$ as the set of ensembles of channels that are pairwise $2\varepsilon$-diamond-far but have $\eta$-close Stinespring isometries:
\begin{equation}
    \mathcal{E}(d_A,d_B,r,\varepsilon,\eta) = 
    \left\{ \{\Phi_i\}_{i=1}^M \subset \mathcal{C}(d_A,d_B,r) 
    \;\middle|\; 
    \begin{array}{l}
        \forall i \neq j: \|\Phi_i - \Phi_j\|_{\diamond} > 2\varepsilon, \\
        \exists \text{ Stinespring isometries } \{V_i\}_{i=1}^M \\
        \text{such that } \forall i \neq j: \|V_i - V_j\|_{\mathrm{op}} \leq \eta
    \end{array}
    \right\}.
\end{equation}

\subsection{The coherent query model}

In the \emph{coherent query model} for quantum channel learning, the learning algorithm is allowed to interleave queries to the unknown channel with arbitrary, adaptively chosen quantum operations. Formally, the  algorithm prepares an initial quantum state $\rho$ on a system comprising the input space $\cH_A$ of dimension $d_A$ together with an auxiliary system $\cH_{\text{aux}}$ of arbitrary dimension. It  then performs $N$ uses of the unknown channel $\Phi$, interspersed with arbitrary quantum channels $\{\cN_i\}_{i=1}^{N-1}$ (the \emph{intermediate operations}) that act jointly on the output space $\cH_B$ and the auxiliary system. The final state after $N$ queries is
\bb
\rho^{\text{output}} = \bigl[\Phi \otimes {\rm Id}_{\text{aux}}\bigr] \circ \cN_{N-1} \circ \bigl[\Phi \otimes {\rm Id}_{\text{aux}}\bigr] \circ \cdots \circ \cN_1 \circ \bigl[\Phi \otimes {\rm Id}_{\text{aux}}\bigr](\rho)\,,
\ee
where each $\Phi \otimes \id_{\text{aux}}$ denotes a query to the unknown channel acting on the input system while leaving the auxiliary system unchanged. Finally, the  algorithm measures $\rho^{\text{output}}$ with a positive operator-valued measure (POVM) to produce a classical description of an estimate $\hat{\Phi}$. The \emph{query complexity} is the minimum  number $N$ of uses of $\Phi$ required to output, with high probability, an estimate $\hat{\Phi}$ such that $\|\Phi - \hat{\Phi}\|_\diamond \le \varepsilon$.

This model  generalizes the parallel (or non-adaptive) query model, in which all $N$ uses of $\Phi$ are applied in parallel on a (possibly entangled) input state, corresponding to the special case where the intermediate operations $\cN_i$ are all identity channels. The coherent model captures the most general physically realizable learning procedure that respects causality and does not assume access to the inverse or conjugate of $\Phi$. It is the natural setting for studying the fundamental quantum limits of channel learning when arbitrary quantum processing between queries is allowed.

\section{A general lower bound on channel learning}\label{sec:LB}

In this section, we prove  a general  lower bound for learning a general quantum channel in diamond distance.

\begin{thm}[(General lower bound)]\label{thm:gen-LB}
    Let $ d_A, d_B, r\ge 1 $, $M\ge 3$ and $\eps, \eta\in (0, 1/2)$. Consider an ensemble $\{\Phi_i\}_{i=1}^M \in  \mathcal{E}(d_A,d_B,r,\varepsilon,\eta)$ of $M$ quantum channels that are $2\varepsilon$-diamond-far and  whose 
Stinespring isometries are $\eta$-operator-norm-close. 
	    Any  coherent algorithm that constructs $\hat{\Phi}_i$ such that  $\|\Phi_i-\hat{\Phi}_i\|_\diamond \le \eps$  with probability at least $2/3$ for all $i\in [M]$ needs at least   
        \bb
        N =  \left\lceil\frac{(2/3)\log(M)-\log 2}{4\eta \log(d_Br/\eta)}\right\rceil
        \ee
        uses of $\cN$.
\end{thm}

We follow a standard strategy for proving lower bounds for learning problems (e.g., ~\cite{flammia2012quantum,haah2017sample,lowe2022lower,fawzi2023lower,oufkir2023sample,Bluhm2024Mar,Rosenthal2024Sep, Mele_2025}). 
\begin{proof}

\begin{figure}[t]
  \centering
  \def\svgwidth{0.62\linewidth}
  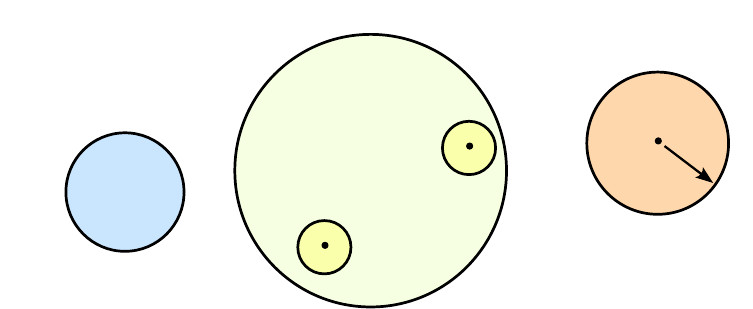
  \caption{Schematic representation of the ensemble $\{\Phi_i\}_{i=1}^M$.} 
  \label{fig:figure2}
\end{figure}

	Let us consider any fixed coherent algorithm that constructs $\hat{\Phi}_i$ such that  $\|\Phi_i-\hat{\Phi}_i\|_\diamond \le \eps$  with probability at least $2/3$ for all $i\in [M]$. Let $X\sim \unif[M]$ and let $Y$ be the index output by such algorithm upon receiving the quantum channel $\Phi_X^{\otimes N}$. Since the quantum channels $\{\Phi_i\}_{i=1}^M$  are pairwise $2\varepsilon$-diamond-far, the algorithm should find $X$ just by picking the $\eps$-closest channel to $\hat{\Phi}_X$ in $\{\Phi_i\}_{i=1}^M$ (see Figure \ref{fig:figure2}). Hence, by Fano's inequality \cite{FANO} we have:
	\bb
		I(X: Y)\ge (2/3)\log(M)-\log 2.
	\ee
 A coherent algorithm using the quantum channel  $\Phi_x(\,\cdot\,)$ chooses the input state $\rho$,  the channels $\cN_{1}, \dots, \cN_{N-1}$, and measures the output state:
\bb
\sigma_{x}^{N} = [\Phi_{x}\otimes {\rm Id}] \circ \cN_{N-1} \circ \cdots \circ \cN_{1}\circ [\Phi_{x}\otimes {\rm Id}](\rho). 
\ee
We can suppose that the channels $\{\Phi_{x}\}_{k\in [N]}$ act on different systems $\{A_k\}_{k\in [N]}$ of dimension $d_A$ (we can include swap channels in $\cN_1, \dots , \cN_{N-1}$ if necessary).
We can assume, without loss of generality, that all the channels $\cN_{1}, \dots, \cN_{N-1}$ are isometries up to modifying the measurement at the end. Similarly, we can suppose that $V_x^{A_k\to B_k} = V_{\Phi_{x}}^{A_k\to B_k}$ is applied instead of $\Phi_{x}$ directly after $\cN_{k-1}$ for $k=1, \dots, N$.  The global system is thus $B_1 \cdots B_N E$ where $|B_k|=d_Br$ and $E$ is an ancilla system of arbitrary dimension. The global state before measurement becomes 
\bb
\sigma_{x}^{N} = [\pazocal{V}_{x}\otimes {\rm Id}] \circ \cN_{N-1} \circ \cdots \circ \cN_{1}\circ [\pazocal{V}_{x}\otimes {\rm Id}](\rho),
\ee
with $\pazocal{V}_x(\,\cdot\,) = V_x(\,\cdot\,)V_x^\dagger$. 
 For $k\in [N]$, we denote $\sigma_{x}^{k} = \pazocal{V}_x \circ \cN_{k-1} \circ \cdots \circ \cN_{1}\circ \pazocal{V}_x(\rho)$, $\sigma_{x}^0 = \rho$ and $\cN_{0}={\rm Id}$, so that we have 
\bb
\sigma_{x}^{k} = \pazocal{V}_x \circ \cN_{k-1}(\sigma_{x}^{k-1}). 
\ee
Denote by $\pi_{k} = \frac{1}{M}\sum_{x=1}^M \pazocal{V}_x \circ \cN_{k-1}(\sigma_{x}^{k-1})$ and $\xi_{k} = \frac{1}{M}\sum_{x=1}^M \pazocal{V}_1\circ \cN_{k-1}(\sigma_{x}^{k-1})$.
Hence the mutual information between $X$ and the observation of the coherent algorithm $Y$ can be bounded as follows
\bb\label{eq:UBmutual-gen}
    I(X:Y)&\leqt{(i)}  S\left(\frac{1}{M}\sum_{x=1}^M \sigma_{x}^N\right)-\frac{1}{M}\sum_{x=1}^MS\left( \sigma_{x}^N\right)
    \\&\eqt{(ii)} \sum_{k=1}^{N} S\left( \frac{1}{M}\sum_{x=1}^M\sigma_{x}^{k}\right)-\sum_{k=1}^{N}S\left( \frac{1}{M}\sum_{x=1}^M\sigma_{x}^{k-1}\right)
    \\&\eqt{(iii)}\sum_{k=1}^{N} S\left(\frac{1}{M}\sum_{x=1}^M \pazocal{V}_x \circ \cN_{k-1}(\sigma_{x}^{k-1})\right)-\sum_{k=1}^{N}S\left( \frac{1}{M}\sum_{x=1}^M \pazocal{V}_1\circ \cN_{k-1}(\sigma_{x}^{k-1})\right)
    \\&=\sum_{k=1}^{N} S\left(B_{1} \cdots B_{k-1}B_{k} A_{k+1}\cdots A_{N} E\right)_{\pi_{k}} -\sum_{k=1}^{N}S\left(B_{1} \cdots B_{k-1}B_{k} A_{k+1}\cdots A_{N} E\right)_{\xi_{k}},
\ee
where (i) uses Holevo's theorem \cite{holevo1973bounds}
; (ii) is a telescopic sum and uses the fact that $S\left( \sigma_{x}^N\right) = S\left( \rho\right)$ for all $x$ as all the applied operations are isometry; (iii) uses the assumption that $\cN_{k-1}$  and $\pazocal{V}_1$ are  isometry channels. 

Now, observe that we have that  $\ptr{B_{k}}{\pi_{k}} = \ptr{B_{k}}{\xi_{k}}$ so we can apply the continuity bound of \cite[Theorem 5]{Berta2024Aug} (see also \cite{Alicki_2004, Winter_2016, Audenaert_2025})
\bb
    &S\left(B_{1} \cdots B_{k-1}B_{k} A_{k+1}\cdots A_{N} E\right)_{\pi_{k}} - S\left(B_{1} \cdots B_{k-1}B_{k} A_{k+1}\cdots A_{N} E\right)_{\xi_{k}}
    \\&\quad = S\left(B_k|B_{1} \cdots B_{k-1} A_{k+1}\cdots A_{N} E\right)_{\pi_{k}} - S\left(B_{k}|B_{1} \cdots B_{k-1} A_{k+1}\cdots A_{N} E\right)_{\xi_{k}}
    \\&\quad \le  \|\pi_{k} - \xi_{k}\|_1\log(|B_{k}|^2) + h_2( \|\pi_{k} - \xi_{k}\|_1)
\ee
with $h_2(a) = -a\log a -(1-a)\log(1-a)$ being the binary entropy. We have that
\bb
    \|\pi_{k} - \xi_{k}\|_1  &= \left\| \frac{1}{M}\sum_{x=1}^M \pazocal{V}_x \circ \cN_{k-1}(\sigma_{x}^{k-1}) -  \frac{1}{M}\sum_{x=1}^M \pazocal{V}_1 \circ \cN_{k-1}(\sigma_{x}^{k-1}) \right\|_1
    \\&\le \frac{1}{M}\sum_{x=1}^M\left\|  (\pazocal{V}_x  - \pazocal{V}_1) \circ \cN_{k-1}(\sigma_{x}^{k-1}) \right\|_1
    \\&=  \frac{1}{M}\sum_{x=1}^M\left\| V_x \zeta V_x^\dagger - V_1 \zeta V_1^\dagger \right\|_1
\ee
 with $\zeta = \cN_{k-1}(\sigma_{x}^{k-1}) $ being a quantum state. Using the triangle inequality, we obtain 
\bb
     \frac{1}{M}\sum_{x=1}^M\left\| V_x \zeta V_x^\dagger - V_1 \zeta V_1^\dagger \right\|_1 &\le  \frac{1}{M}\sum_{x=1}^M\left(\left\| (V_x-V_1) \zeta V_x^\dagger\right\|_1 + \left\| V_1 \zeta (V_x-V_1)^\dagger\right\|_1\right)
     \\&\le  \frac{1}{M}\sum_{x=1}^M\left( \|V_x-V_1\|_{\mathrm{op}}\left\|  \zeta V_x^\dagger\right\|_1+ \|V_x-V_1\|_{\mathrm{op}}\left\| V_1\zeta\right\|_1\right)
     \\&\le 2\eta.
\ee
Therefore, we deduce 
\bb
    I(X:Y)&\le \sum_{k=1}^{N} S\left(B_{1} \cdots B_{k-1}B_{k} A_{k+1}\cdots A_{N} E\right)_{\pi_{k}} -\sum_{k=1}^{N}S\left(B_{1} \cdots B_{k-1}B_{k} A_{k+1}\cdots A_{N} E\right)_{\xi_{k}}
    \\&\le \sum_{k=1}^{N}\|\pi_{k} - \xi_{k}\|_1\log(|B_{k}|^2) + h_2( \|\pi_{k} - \xi_{k}\|_1)
      \\&\le \sum_{k=1}^{N} 2\eta \log((d_Br)^2) + h_2( 2\eta)
    \\&\le 4N\eta \log(d_Br/\eta),
\ee
where we used that $h_2(a) \le 2a \log(1/a)$ for $a\in (0, \frac{1}{2})$.
Since $I(X:Y)\ge (2/3)\log(M)-\log2$ we deduce that 
\bb
    N \ge \frac{(2/3)\log(M)-\log 2}{4\eta \log(d_Br/\eta)}. 
\ee
This concludes the proof.
\end{proof}

Given Theorem \ref{thm:gen-LB}, we can prove lower bounds on learning quantum channels by constructing an ensemble $\{\Phi_i\}_{i=1}^M$ within $\mathcal{E}(d_A,d_B,r,\varepsilon,\eta)$ containing $M$ quantum channels that are pairwise $2\varepsilon$-diamond-far, yet whose Stinespring isometries are pairwise $\eta$-operator-norm-close. The resulting lower bound then scales with $M$ and inversely with $\eta$. To strengthen this bound, we should aim to construct an ensemble that maximizes $M$ while minimising $\eta$\footnote{Note that by the inequality  $\|\Phi_x - \Phi_y\|_\diamond \le \|V_{x} - V_y\|_{\mathrm{op}}$ \cite{kretschmann2008information}, the parameter $\eta$ should be at least $2\eps$.}.

A natural approach to constructing such an ensemble is to use existing packing nets. However, this leads to an ensemble in $\mathcal{E}(d_A,d_B,r,\varepsilon,4\sqrt{\eps})$ of cardinality $M$ satisfying, $\log M = \Omega(d_Ad_B r)$ and by Theorem \ref{thm:gen-LB} implies  the following weak lower bound  (see Appendix \ref{sec:LB-sqrt(eps)}  for details):
\begin{align}\label{eq:sqrt}
    N \ge  \Omega\left(\frac{d_Ad_Br}{\sqrt{\eps}\log(d_Br/\sqrt{\eps})}\right). 
\end{align}
In what follows, we improve the $\varepsilon$-dependence of this lower bound by constructing a new  ensemble in $\mathcal{E}(d_A,d_B,r,\varepsilon,\eta)$ with comparable cardinality but with $\eta = O(\varepsilon)$ rather than $O(\sqrt{\varepsilon})$.

\newpage

\section{Improved lower bounds}\label{sec:improved}

\begin{figure}[t]
  \centering
  \def\svgwidth{\linewidth}
  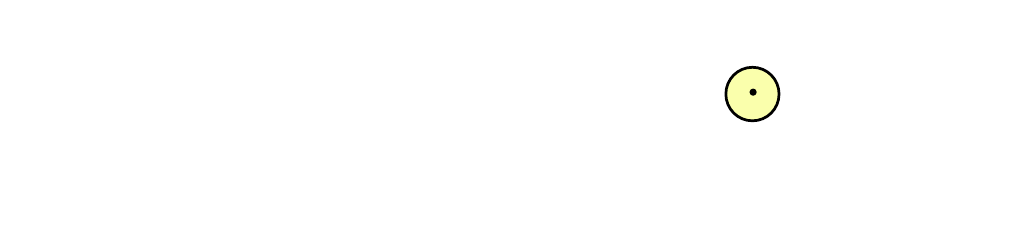
  \caption{Schematic construction of the isometries $V_x$ and of the channels $\Phi_x$, where $U_x\sim \Haar({\rm U}(rd_B)) $, $\widetilde{V}_x=U_xS$, $V_x =\sqrt{1-\eps^2} \ket{0}\otimes \widetilde{V}_0 + \eps \ket{1}\otimes \widetilde{V}_x$, $\Phi_x(\,\cdot\,)=\Tr_E[\widetilde{V}_x\,\cdot\, \widetilde{V}_x^\dagger]$.} 
  \label{fig:figure}
\end{figure}
In this section, we improve the bound \eqref{eq:sqrt} by constructing an ensemble  $\mathcal{E}(d_A,d_B,r,\Omega(\eps) , 2\eps)$ with cardinality $M$ satisfying $\log M = \Omega(d_Ad_B r)$. More precisely, we construct a set of isometries $\{{V}_x\}_{x\in [M]}$ corresponding to quantum channels $\{\Phi_x\}_{x\in [M]}$ such that  $\|\Phi_x - \Phi_y\|_\diamond \ge \Omega(\eps)$,  $  \| {V}_x -  {V}_y\|_{\rm{op}} \le O(\eps)$, and $\log M = \Omega(d_Ad_B r)$, as in Figure \ref{fig:figure}. We prove the existence of  such a set  using a probabilistic argument and distinguish between the case $d_A=rd_B$ and $d_A\le rd_B/2$. \\

\textbf{Case $d_A =rd_B$.} Let  $\theta\in(\pi/2,\pi]$ such that $\epsilon =-2\cos\theta$. Let us define the $d_A\times d_A$ matrix
\bb
    O\coloneqq\begin{cases}
        \eps\cdot {\rm diag}( \mathrm{e}^{i\theta}, \mathrm{e}^{-i\theta},\dots,  \mathrm{e}^{i\theta}, \mathrm{e}^{-i\theta}) & \text{if $d_A$ is even,}\\
        \eps\cdot {\rm diag}( \mathrm{e}^{i\theta}, \mathrm{e}^{-i\theta},\dots,  \mathrm{e}^{i\theta}, \mathrm{e}^{-i\theta},0) & \text{if $d_A$ is odd.}
    \end{cases}
\ee
For $x \in[M]$,  let $U_x\sim \Haar({\rm U}(d_Br))$ and we  define the Stinespring isometry
\bb\label{construction-eq}
V_x \coloneqq U_x(\id + O)U_x^\dagger,
\ee 
and let ${\Phi}_x(\,\cdot\,) \coloneqq \ptr{E}{{V}_x\,\cdot\,{V}_x^\dagger}$ the corresponding quantum channel. It has Kraus rank at most $|E| = r$.  The quantum channel ${\Phi}_x$ has input system $A$ of dimension $d_A$ and output system $B$ of dimension $d_B$.\\

\textbf{Case $d_A\le rd_B/2$.}
Let $\widetilde{\Phi}_0$ be a quantum channel  with Kraus operators $\{K_{0,i}\}_{i\in [r]}$ satisfying
\begin{align}\label{eq:Kraus-1}
    \left|\Tr\left[ K_{0,i}^\dagger  K_{0,j}\right]\right| \le \frac{2d_A}{r}\delta_{i,j}\,, \quad \forall i, j\in [r]. 
\end{align}
The existence of such a channel is shown in Appendix \ref{sec:existence}. Let  $\widetilde{V}_0^{A\to BE} = \sum_{i=1}^r \ket{i}_E \otimes  K_{0,i}$ be a Stinespring isometry of the quantum channel $\widetilde{\Phi}_0$. For $x \in[M]$,  let $U_x\sim \Haar({\rm U}(d_Br))$ and define
\bb\label{eq:V}
    S^{A\to BE}&\coloneqq \sum_{i=1}^{d_A} \ket{i}_{BE}\bra{i}_A \in \mathbb{C}^{rd_B\times d_A} \qquad \text{and}\qquad  \widetilde{V}_x^{A\to BE} \coloneqq U_x^{BE}S^{A\to BE}. 
\ee
Then we define the isometries 
\bb\label{construction-leq}
    V^{A\to FBE}_x \coloneqq \sqrt{1-\eps^2} \ket{0}_F\otimes \widetilde{V}_0^{A\to BE} + \eps \ket{1}_F\otimes \widetilde{V}_x^{A\to BE}
\ee
and we call ${\Phi}_x(\,\cdot\,) \coloneqq \ptr{E}{{V}_x\,\cdot\,{V}_x^\dagger}$ the corresponding quantum channel, which has Kraus rank at most $|E| = r$.  The quantum channel ${\Phi}_x$ has input system $A$ of dimension $d_A$ and output system $FB$ of dimension $2d_B$.\\

For both constructions, we have the following technical lemma. 
\begin{lemma}\label{lem:ppts-Phi}
 Let us suppose that either $d_A = rd_B$ or $d_A\le rd_B/2$ and, given any unitary $U_i\in{\rm U}(rd_B)$, let $\Phi_i$ be the channel constructed from according to \eqref{construction-eq} or according to \eqref{eq:V} and \eqref{construction-leq}, respectively. Then, calling $J_{\Phi_i}$ the Choi state of the channel $\Phi_i$, the function 
 \bb
 f:(U_1, U_2)\in {\rm U}(rd_B)^2 \mapsto \|J_{\Phi_1} - J_{\Phi_2} \|_1
 \ee
is $L\coloneqq 4\sqrt{\frac{2}{d_A}}\eps$-Lipschitz with respect to the $\ell_2$-sum of the 2-norms, namely
\bb
    |f(U_1,U_2)-f(U'_1,U'_2)|\leq  L  \|(U_1,U_2) -(U_1',U_2')\|_2
\ee
for all $U_1,U_1',U_2,U_2'\in {\rm U}(rd_B)$, where $\|(A,B)\|_2\coloneqq\sqrt{\|A\|_2^2+\|B\|_2^2}$. 
Furthermore, if we consider independent random unitaries $U_1,U_2 \sim \Haar({\rm U}(rd_B))$, we have
\bb\label{eq:LB-1st moment-uni}
\ex{\|J_{\Phi_1} - J_{\Phi_2}\|_1}\ge  \Omega(\eps). 
\ee
\end{lemma}
 \begin{proof}
   The proof of this lemma is deferred to Section \ref{app:first-moment-and-lip}.
\end{proof}
Since the function $f$ is Lipschitz it concentrates around its mean by the following theorem. 
\begin{thm}[{\cite[Corollary~17]{meckes2013spectral}}]\label{lem:corollary17} Let $k,d\geq 1$. Suppose that $f: \big( {\rm U}(d)\big)^k\to \mathbb{R}$ is $L$-Lipschitz with respect to the $\ell_2$-sum of the 2-norms, i.e.
\bb
    \big|f(U_1,\dots, U_k)-f(U'_1,\dots, U'_k)\big|\leq L\sqrt{\sum_{i=1}^k\|U_i-U_i'\|_2^2}
\ee
for all $U_i,U_i'\in {\rm U}(d)$, with $i=1,\dots, k$. Then, if we independently sample $U_1,\dots, U_k$ according to the Haar measure on ${\rm U}(d)$, the following inequality holds for each $t>0$:
\bb
    \prob{f(U_1,\dots, U_k)  \geq \ex{f(U_1,\dots, U_k)} +t} \le  \exp\left(-\frac{dt^2}{12L^2}\right).
\ee
\end{thm}

We now have all the ingredients to prove the  main result.

\begin{boxedthm}{}
    \begin{thm}[(Improved lower bound for channel learning)]\label{thm:main}
    Let $\eps \in (0, 10^{-4})$, $ d_A, d_B, r\ge 1$ such that  either 
    $d_A = rd_B$ or $d_A\le rd_B/2$ and
    $d_Ad_Br \ge 2500$.  
	    Any  coherent algorithm that constructs $\hat{\cN}$ such that  $\|\cN-\hat{\cN}\|_\diamond \le \eps$  with probability at least $2/3$ needs  at least  a number of channel uses satisfying 
        \bb
            N = \begin{cases}
                \Omega\left(\frac{d_Ad_Br}{\eps \log(d_Br/\eps)} \right) &\text{ when } \quad d_A = rd_B, \\ 
                 \Omega\left(\frac{d_Ad_Br}{\eps^2 \log(d_Br/\eps)} \right) & \text{ when }\quad d_A\le rd_B/2.
            \end{cases}
        \ee
\end{thm}
\end{boxedthm}

In the case $d_A\le rd_B/2$, the  lower bound  is near-optimal given the general upper bounds of \cite{AMele2025,chen2025quantumchanneltomographyestimation}. 

\begin{proof}
In both constructions \eqref{construction-eq} and \eqref{construction-leq}, the Stinespring isometries are $2\eps$ close in operator norm. 
In the case $d_A\le rd_B/2$,  the construction of the isometries $\{{V}^{A\to BE}_x\}_{x\in [M]}$ in \eqref{construction-leq} satisfies

\bb
    \| {V}_x -  {V}_y\|_{\rm{op}} &=  \|\eps \ket{1}\otimes \widetilde{V}_x -  \eps \ket{1}\otimes \widetilde{V}_y\|_{\rm{op}} = \eps  \| \widetilde{V}_x -  \widetilde{V}_y\|_{\rm{op}} \le 2\eps. 
\ee
In the case $d_A=rd_B$, the construction of the isometries $\{{V}^{A\to BE}_x\}_{x\in [M]}$ in \eqref{construction-eq} satisfies
\bb
\|V_x -V_y\|_{\rm{op}}& = \|U_xO U_x^\dagger -U_yO U_y^\dagger\|_{\rm{op}}
\le \|U_xO U_x^\dagger\|_{\rm{op}} +\|U_yO U_y^\dagger\|_{\rm{op}}
= 2\|O\|_{\rm{op}}
= 2\eps. 
\ee

Furthermore, since the function $f:(U_x,U_y) \mapsto \|\Phi_x - \Phi_y\|_1$ is 
$L$-Lipschitz, when we sample two independent unitaries $U_x,U_y\sim \Haar({\rm U}(rd_B))$, by Theorem \ref{lem:corollary17}, we have\footnote{To be precise, we are applying Theorem \ref{lem:corollary17} to $-f$, which is also $L$-Lipschitz.}
\begin{align}
    \prob{\ex{ f(U_x,U_y)}- f(U_x,U_y) \geq \frac{1}{2} \ex{ f(U_x,U_y)}} \le  \exp\left(-\frac{d_Br\eps^2}{12 L^2}\right)  = \mathrm{e}^{-cd_Ad_Br} \eqcolon \delta
\end{align}
with $c$ a universal constant. 
Therefore, with probability at least $1- \delta$, we have 
\begin{align}
    f(U_x,U_y) >  \frac{1}{2} \ex{ f(U_x,U_y)} \ge  \Omega(\eps)
\end{align}
where we have used the lower bound \eqref{eq:LB-1st moment-uni}. Let 
\bb
    M \coloneqq \left\lfloor\exp\left(\tfrac{c}{2}d_Ad_Br-1\right)\right\rfloor
\ee
and let $\{U_x\}_{x\in [M]}$ be i.i.d.\ Haar random matrices. Note that, by the very definition of $M$ and $\delta$, we have $M^2\delta<1$ and $\log M = \Omega(d_Ad_Br)$. 
By the union bound,

\bb
     &\prob{\exists x\neq y : f(U_x,U_y) < \frac{1}{2}\ex{f(U_x,U_y)} }\\
     &\qquad \qquad \le M(M-1)\; \prob{ f(U_x,U_y) < \frac{1}{2}\ex{f(U_x,U_y)} }\le M^2 \delta<1. 
\ee
Hence, there exists a family $\{U_x\}_{x\in [M]}$ such that, for all $x\neq y$, 
\begin{align}
     \left\|{\Phi}_x -{\Phi}_y\right\|_{\diamond} = f(U_x,U_y) \ge  \frac{1}{2}\ex{f(U_x,U_y)}  \ge \Omega(\eps).
\end{align}
 Hence, we have provided a construction of an ensemble  $\{\Phi_x\}_{x\in[M]} \in \mathcal{E}(d_A, d_B, r, \Omega(\eps), 2\eps)$ of cardinality $\log M = \Omega(d_Ad_Br)$.
By Theorem \ref{thm:gen-LB}, we  conclude:
 \bb
            N \ge  \frac{(2/3)\log(M)-\log2}{4\eta \log(d_Br/\eta)} \ge\Omega\left(  \frac{d_Ad_Br  }{ \eps \cdot  \log(d_Br/\eps)} \right), 
        \ee
    which completes the proof for the case $d_A = rd_B$. 

    For the remaining case $d_A\le rd_B/2$, we return to the proof of Theorem~\ref{thm:gen-LB} (precisely \eqref{eq:UBmutual-gen}) and improve the upper bound  on  the Holevo information using the specificities of the construction \eqref{construction-leq}.
  
    Recall that, from the proof of Theorem \ref{thm:gen-LB}, the mutual information can be bounded as follows \eqref{eq:UBmutual-gen}
\bb\label{eq:UBmutualI}
    &I(X:Y)\\
    &\quad \le  \sum_{k=1}^{N} S\left(\frac{1}{M}\sum_{x=1}^M \pazocal{V}_x \circ \cN_{k-1}(\sigma_{x}^{k-1})\right)-\sum_{k=1}^{N}S\left( \frac{1}{M}\sum_{x=1}^M \pazocal{V}_1\circ \cN_{k-1}(\sigma_{x}^{k-1})\right)
    \\&\quad =\sum_{k=1}^{N} S\left( \tilde{B}_{1} \cdots \tilde{B}_{k-1} F_kB_{k} E_k A_{k+1}\cdots A_{N} E\right)_{\pi_{k}} -\sum_{k=1}^{N}S\left(\tilde{B}_{1} \cdots \tilde{B}_{k-1}F_kB_{k}E_k A_{k+1}\cdots A_{N} E\right)_{\xi_{k}},
\ee
    with $\tilde{B}_{k} = F_kB_kE_k$, $\pazocal{V}_x^{A_k\to \tilde{B}_k}(\,\cdot\,) = V_x^{A_k\to \tilde{B}_k}\,\cdot\,(V_x^{A_k\to \tilde{B}_k})^ \dagger, $
    \bb
    \pi_{k} = \frac{1}{M}\sum_{x=1}^M \pazocal{V}_x \circ \cN_{k-1}(\sigma_{x}^{k-1}) \quad \text{and}\quad  \xi_{k} = \frac{1}{M}\sum_{x=1}^M \pazocal{V}_1\circ \cN_{k-1}(\sigma_{x}^{k-1}).
    \ee 
    Observe that the construction  \eqref{construction-leq} implies that, calling $\rho^{FBE} = V_x^{A\to FBE} \phi (V_x^{A\to FBE})^\dagger$ (where $\phi$ is any arbitrary state in $\mathcal{H}_A$) we have
\bb\label{eq:partial_traces}
     \ptr{F}{\rho}&= (1-\eps^2)  \widetilde{V}_0 \phi \widetilde{V}_0^\dagger+ \eps^2  \widetilde{V}_x \phi \widetilde{V}_x^\dagger,\\
    \ptr{BE}{\rho} &=  \begin{pmatrix}
1 - \eps^2 & \sqrt{1-\eps^2}\eps \Tr[\widetilde{V}_0 \phi \widetilde{V}_x^\dagger]\\
\sqrt{1-\eps^2}\eps \Tr[\widetilde{V}_x \phi \widetilde{V}_0^\dagger] & \eps^2
\end{pmatrix}.
\ee

For any bipartite state $\zeta_{AB}$, the von Neumann entropy satisfies the subadditivity property   $ S(AB)_{\zeta}\le S(A)_{\zeta}+S(B)_{\zeta}$ and the triangle inequality $|S(A)_{\zeta} - S(B)_{\zeta}|\le S(AB)_{\zeta}$ \cite{Nielsen2010Dec}, so we obtain  for all $k\in [N]$:
\bb\label{eq:entropy}
&S\left(\tilde{B}_{1} \cdots \tilde{B}_{k-1} F_kB_{k} E_k A_{k+1}\cdots A_{N} E\right)_{\pi_{k}}\\
&\qquad\qquad  \le S\left(\tilde{B}_{1} \cdots \tilde{B}_{k-1} B_{k} E_k A_{k+1}\cdots A_{N} E\right)_{\pi_{k}} + S(F_{k})_{\pi_{k}},\\[0.5em]
-& S\left(\tilde{B}_{1} \cdots \tilde{B}_{k-1} F_kB_{k} E_k A_{k+1}\cdots A_{N} E\right)_{\xi_{k}}\\
&\qquad \qquad \le -S\left(\tilde{B}_{1} \cdots \tilde{B}_{k-1} B_{k} E_k A_{k+1}\cdots A_{N} E\right)_{\xi_{k}}  +S\left(F_{k} \right)_{\xi_{k}}.
\ee
Therefore 
\bb\label{eq:UBmutualII}
    &S\left( \tilde{B}_{1} \cdots \tilde{B}_{k-1} F_kB_{k} E_k A_{k+1}\cdots A_{N} E\right)_{\pi_{k}} -S\left(\tilde{B}_{1} \cdots \tilde{B}_{k-1} F_kB_{k}E_k A_{k+1}\cdots A_{N} E\right)_{\xi_{k}}
    \\& \le 2S\left( F_k \right)_{\pi_{k}} +S\left( \tilde{B}_{1} \cdots \tilde{B}_{k-1} B_{k} E_k A_{k+1}\cdots A_{N} E\right)_{\pi_{k}} -S\left(\tilde{B}_{1} \cdots \tilde{B}_{k-1}B_{k}E_k A_{k+1}\cdots A_{N} E\right)_{\xi_{k}}
    \\&= 2S\left( F_k \right)_{\pi_{k}} +S\left( B_{k} E_k|  \tilde{B}_{1} \cdots \tilde{B}_{k-1} A_{k+1}\cdots A_{N} E\right)_{\pi_{k}} -S\left( B_{k} E_k| \tilde{B}_{1} \cdots \tilde{B}_{k-1} A_{k+1}\cdots A_{N} E\right)_{\xi_{k}}.
\ee
Now, observe that we have that  $\ptr{B_{k}E_kF_k}{\pi_{k}} = \ptr{B_{k}E_kF_k}{\xi_{k}}$ so we can apply the continuity bound of \cite[Theorem 5]{Berta2024Aug} (see also \cite{Alicki_2004, Winter_2016, Audenaert_2025})
\bb
    & S\left(B_kE_k|\tilde{B}_{1} \cdots \tilde{B}_{k-1} A_{k+1}\cdots A_{N} E\right)_{\pi_{k}} - S\left(B_kE_k|\tilde{B}_{1} \cdots \tilde{B}_{k-1} A_{k+1}\cdots A_{N} E\right)_{\xi_{k}}
    \\&\quad \le  \|\ptr{F_k}{\pi_{k}} - \ptr{F_k}{\xi_{k}}\|_1\log(|B_{k}E_k|^2) + h_2( \|\ptr{F_k}{\pi_{k}} - \ptr{F_k}{\xi_{k}}\|_1). 
\ee
Remark that, calling $\tilde{\sigma}_x^{k-1} \coloneqq \pazocal{N}_{k-1}(\sigma_x^{k-1})$, by \eqref{eq:partial_traces} we have
\bb
 \left\|\ptr{F_k}{\pi_{k}} - \ptr{F_k}{\xi_{k}}\right\|_1 &=\bigg\|\frac{1}{M}\sum_{x=1}^M \left(  (1-\eps^2)  \widetilde{V}_0 \tilde{\sigma}_x^{k-1}  \widetilde{V}_0^\dagger + \eps^2  \widetilde{V}_x \tilde{\sigma}_x^{k-1} \widetilde{V}_x^\dagger\right)
 \\&\qquad -\frac{1}{M}\sum_{x=1}^M  \left((1-\eps^2)  \widetilde{V}_0 \tilde{\sigma}_x^{k-1}  \widetilde{V}_0^\dagger + \eps^2  \widetilde{V}_1 \tilde{\sigma}_x^{k-1} \widetilde{V}_1^\dagger\right)\bigg\|_1
 \\&=\left\| \eps^2 \frac{1}{M}\sum_{x=1}^M (\widetilde{V}_x \tilde{\sigma}_x^{k-1}  \widetilde{V}_x^\dagger-\widetilde{V}_1 \tilde{\sigma}_x^{k-1}  \widetilde{V}_1^\dagger)  \right\|_1
 \\&\le 2\eps^2. 
\ee
Hence 
\bb\label{eq:UBmutualIII}
    & S\left(B_kE_k|\tilde{B}_{1} \cdots \tilde{B}_{k-1} A_{k+1}\cdots A_{N} E\right)_{\pi_{k}} - S\left(B_kE_k|\tilde{B}_{1} \cdots \tilde{B}_{k-1} A_{k+1}\cdots A_{N} E\right)_{\xi_{k}}
    \\&\quad \le  \|\ptr{F_k}{\pi_{k}} - \ptr{F_k}{\xi_{k}}\|_1\log(|B_{k}E_k|^2) + h_2( \|\ptr{F_k}{\pi_{k}} - \ptr{F_k}{\xi_{k}}\|_1)
    \\&\quad \le 8\eps^2\log(d_Br/\eps). 
\ee
Moreover, by applying  the data processing inequality for  the  dephasing channel, we obtain 
\bb\label{eq:UBmutualIV}
S\left(F_k\right)_{\pi_{k}}
&= S\left( \frac{1}{M}\sum_{x=1}^M \begin{pmatrix}
1 - \eps^2 & \sqrt{1-\eps^2}\eps \Tr[\widetilde{V}_0 \tilde{\sigma}_x^{k-1} \widetilde{V}_x^\dagger]\\
\sqrt{1-\eps^2}\eps \Tr[\widetilde{V}_x \tilde{\sigma}_x^{k-1}\widetilde{V}_0^\dagger] & \eps^2
\end{pmatrix}\right)
\\&\le 
S\left( \frac{1}{M}\sum_{x=1}^M \begin{pmatrix}
1 - \eps^2 & 0\\
0 & \eps^2
\end{pmatrix}\right)
\\&= h_2(\eps^2)
\\&\le 4 \eps^2\log(1/\eps). 
\ee
Finally combining \eqref{eq:UBmutualI}, \eqref{eq:UBmutualII}, \eqref{eq:UBmutualIII} and \eqref{eq:UBmutualIV}, we get
\bb
  I(X:Y)&\le \sum_{k=1}^N 2S\left( F_k \right)_{\pi_{k}}+S\left( B_{k} E_k|  \tilde{B}_{1} \cdots \tilde{B}_{k-1} A_{k+1}\cdots A_{N} E\right)_{\pi_{k}} \\[-1.5em]
  &\qquad\phantom{\sum_{k=1}^N 2S\left( F_k \right)_{\pi_{k}}}  -S\left( B_{k} E_k| \tilde{B}_{1} \cdots \tilde{B}_{k-1} A_{k+1}\cdots A_{N} E\right)_{\xi_{k}}
    \\[-0.5em]&\le \sum_{k=1}^N (8 \eps^2\log(1/\eps) + 8\eps^2\log(d_Br/\eps))
    \\&\le  16N\eps^2\log(d_Br/\eps).
\ee
Since $I(X:Y)\ge (2/3)\log(M)-\log2 \ge \Omega(d_Ad_Br)$ we deduce that 
\bb
    N \ge \Omega\left( \frac{d_Ad_Br}{\eps^2\log(d_Br/\eps)}\right). 
\ee
This concludes the proof.
\end{proof}

\section{Conclusion}

We have proved that learning an unknown quantum channel to diamond distance $\varepsilon$ requires $\widetilde{\Omega}\bigl(d_A d_B r / \varepsilon \bigr)$ queries when  $d_A=rd_B$ and $\widetilde{\Omega}\bigl(d_A d_B r / \varepsilon^2\bigr)$ queries when  $d_A\le rd_B/2$, improving upon the previous $\Omega(d_A d_B r)$ bound \cite{Girardi2025Dec}.  So the query complexity $N$ is characterised   by $N=\widetilde{\Theta}\bigl(d_A d_B r / \varepsilon^2 \bigr)$ when $d_A\le rd_B/2$ \cite{AMele2025,chen2025quantumchanneltomographyestimation} and satisfies $\widetilde{\Omega}\bigl(d_A d_B r / \varepsilon \bigr)\le N\le O\bigl(\min\{d_A^{2.5}/\eps, d_A^2/\eps^2\} \bigr)$ when $d_A=rd_B$ \cite{chen2025quantumchanneltomographyestimation}.  

It remains an interesting open problem to determine the optimal query complexity of channel learning in diamond distance for the parameter ranges $r \in  [\frac{d_A}{d_B},\frac{2d_A}{d_B})$, and to investigate the role of quantum memory in the query complexity.

\subsection*{Acknowledgments}
FG acknowledges financial support from the European Union (ERC StG ETQO, Grant Agreement no.\ 101165230).
\bibliography{bib}

@article{Winter_2016,
   title={Tight Uniform Continuity Bounds for Quantum Entropies: Conditional Entropy, Relative Entropy Distance and Energy Constraints},
   volume={347},
   ISSN={1432-0916},
   url={http://dx.doi.org/10.1007/s00220-016-2609-8},
   DOI={10.1007/s00220-016-2609-8},
   number={1},
   journal={Communications in Mathematical Physics},
   publisher={Springer Science and Business Media LLC},
   author={Winter, Andreas},
   year={2016},
   month=mar, pages={291–313} }

@article{Alicki_2004,
   title={Continuity of quantum conditional information},
   volume={37},
   ISSN={1361-6447},
   url={http://dx.doi.org/10.1088/0305-4470/37/5/L01},
   DOI={10.1088/0305-4470/37/5/l01},
   number={5},
   journal={Journal of Physics A: Mathematical and General},
   publisher={IOP Publishing},
   author={Alicki, R and Fannes, M},
   year={2004},
   month=jan, pages={L55–L57} }

@article{Audenaert_2025,
   title={Continuity Bounds for Quantum Entropies Arising From a Fundamental Entropic Inequality},
   volume={71},
   ISSN={1557-9654},
   url={http://dx.doi.org/10.1109/TIT.2025.3586478},
   DOI={10.1109/tit.2025.3586478},
   number={9},
   journal={IEEE Transactions on Information Theory},
   publisher={Institute of Electrical and Electronics Engineers (IEEE)},
   author={Audenaert, Koenraad and Bergh, Bjarne and Datta, Nilanjana and Jabbour, Michael G. and Capel, Ángela and Gondolf, Paul},
   year={2025},
   month=sep, pages={7029–7038} }

@book{Nielsen2010Dec,
	author = {Nielsen, Michael A. and Chuang, Isaac L.},
	title = {{Quantum Computation and Quantum Information: 10th Anniversary Edition}},
	journaltitle = {Higher Education from Cambridge University Press},
	year = {2010},
	publisher = {Cambridge University Press}
}

@article{Mele_2025,
   title={Learning quantum states of continuous-variable systems},
   volume={21},
   ISSN={1745-2481},
   url={http://dx.doi.org/10.1038/s41567-025-03086-2},
   DOI={10.1038/s41567-025-03086-2},
   number={12},
   journal={Nature Physics},
   publisher={Springer Science and Business Media LLC},
   author={Mele, Francesco A. and Mele, Antonio A. and Bittel, Lennart and Eisert, Jens and Giovannetti, Vittorio and Lami, Ludovico and Leone, Lorenzo and Oliviero, Salvatore F. E.},
   year={2025},
   month=nov, pages={2002–2008} }

@article{collins2006integration,
  title={Integration with respect to the Haar measure on unitary, orthogonal and symplectic group},
  author={Collins, Beno{\^ i}t and {\'S}niady, Piotr},
  journal={Communications in Mathematical Physics},
  volume={264},
  number={3},
  pages={773--795},
  year={2006},
  publisher={Springer},
  doi={10.1007/s00220-006-1554-3}
}

@inproceedings{haah2023query,
  title={Query-optimal estimation of unitary channels in diamond distance},
  author={Haah, Jeongwan and Kothari, Robin and O’Donnell, Ryan and Tang, Ewin},
  booktitle={2023 IEEE 64th Annual Symposium on Foundations of Computer Science (FOCS)},
  pages={363--390},
  year={2023},
  organization={IEEE}
}

@article{kretschmann2008information,
  title={The information-disturbance tradeoff and the continuity of Stinespring's representation},
  author={Kretschmann, Dennis and Schlingemann, Dirk and Werner, Reinhard F},
  journal={IEEE transactions on information theory},
  volume={54},
  number={4},
  pages={1708--1717},
  year={2008},
  publisher={IEEE}
}

@book {FANO,
  author = {Fano, R. M.},
  title = {Transmission of information: {A} statistical theory of communications},
  publisher = {The M.I.T. Press, Cambridge, Mass.; John Wiley \& Sons, Inc., New York-London},
  year = {1961},
  pages = {x+389}
}

@article{chen2025quantumchanneltomographyestimation,
	author = {Chen, Kean and Yu, Nengkun and Zhang, Zhicheng},
	title = {{Quantum channel tomography and estimation by local test}},
	journal = {arXiv},
	year = {2025},
	month = dec,
	date = {2025-12-15},
	urldate = {2026-01-07},
	eprint = {2512.13614},
	doi = {10.48550/arXiv.2512.13614}
}

@article{AMele2025,
	author = {Mele, Antonio Anna and Bittel, Lennart},
	title = {{Optimal learning of quantum channels in diamond distance}},
	journal = {arXiv},
	year = {2025},
	month = dec,
	date = {2025-12-11},
	urldate = {2026-01-07},
	eprint = {2512.10214},
	doi = {10.48550/arXiv.2512.10214}
}

@techreport{gu2013moments,
  title       = {Moments of random matrices and Weingarten functions},
  author      = {Gu, Yinzheng},
  year        = {2013},
  institution = {Queen's University},
  url         = {https://qspace.library.queensu.ca/server/api/core/bitstreams/cee37ba4-2035-48e0-ac08-2974e082a0a9/content},
  type        = {Technical report}
}

@article{Rosenthal2024Sep,
	author = {Rosenthal, Gregory and Aaronson, Hugo and Subramanian, Sathyawageeswar and Datta, Animesh and Gur, Tom},
	title = {{Quantum Channel Testing in Average-Case Distance}},
	journal = {arXiv},
	year = {2024},
	month = sep,
	date = {2024-09-19},
	urldate = {2026-01-07},
	eprint = {2409.12566},
	doi = {10.48550/arXiv.2409.12566}
}

@article{holevo1973bounds,
  title={Bounds for the quantity of information transmitted by a quantum communication channel},
  author={Holevo, Alexander Semenovich},
  journal={Problems of Information Transmission},
  volume={9},
  pages={177–183},
  year={1973}
}

@article{flammia2012quantum,
	title={Quantum tomography via compressed sensing: error bounds, sample complexity and efficient estimators},
	author={Flammia, Steven T. and Gross, David and Liu, Yi-Kai and Eisert, Jens},
	journal={New Journal of Physics},
	volume={14},
	number={9},
	pages={095022},
	year={2012},
	publisher={IOP Publishing},
doi={10.1088/1367-2630/14/9/095022}
}

@inproceedings{oufkir2023sample,
  title={Sample-Optimal Quantum Process Tomography with non-adaptive Incoherent Measurements},
  author={Oufkir, Aadil},
  booktitle={2023 IEEE International Symposium on Information Theory (ISIT'23)},
  year={2023},
  doi = {10.1109/ISIT54713.2023.10206538},
}

@article{Bluhm2024Mar,
	author = {Bluhm, Andreas and Caro, Matthias C. and Oufkir, Aadil},
	title = {{Hamiltonian Property Testing}},
	journal = {arXiv},
	year = {2024},
	month = mar,
	date = {2024-03-05},
	urldate = {2026-01-05},
	eprint = {2403.02968},
	doi = {10.48550/arXiv.2403.02968}
}

@article{fawzi2023lower,
	author = {Fawzi, Omar and Oufkir, Aadil and Franca, Daniel Stilck},
	title = {{Lower Bounds on Learning Pauli Channels with Individual Measurements}},
	journal = {IEEE Transactions on Information Theory},
	shortjournal = {IEEE Trans. Inf. Theory},
    year = {2025},
	publisher = {IEEE}
}

@article{lowe2022lower,
author = {Lowe, Angus and Nayak, Ashwin},
title = {Lower Bounds for Learning Quantum States with Single-Copy Measurements},
year = {2025},
issue_date = {March 2025},
publisher = {Association for Computing Machinery},
address = {New York, NY, USA},
volume = {17},
number = {1},
issn = {1942-3454},
url = {https://doi.org/10.1145/3717450},
doi = {10.1145/3717450},
journal = {ACM Trans. Comput. Theory},
month = mar,
articleno = {7},
numpages = {42}
}

@article{haah2017sample,
  title={Sample-Optimal Tomography of Quantum States},
  author={Haah, Jeongwan and Harrow, Aram W. and Ji, Zhengfeng and Wu, Xiaodi and Yu, Nengkun},
  journal={IEEE Transactions on Information Theory},
  volume={63},
  number={9},
  pages={5628--5641},
  year={2017},
  publisher={IEEE},
  doi={10.1109/TIT.2017.2719044}
}

@article{Berta2024Aug,
	author = {Berta, Mario and Lami, Ludovico and Tomamichel, Marco},
	title = {{Continuity of entropies via integral representations}},
	journal = {arXiv},
	year = {2024},
	month = aug,
	date = {2024-08-27},
	urldate = {2025-12-24},
	eprint = {2408.15226},
	doi = {10.1109/TIT.2025.3527858}
}

@article{meckes2013spectral,
	title = {Spectral measures of powers of random matrices},
	volume = {18},
	doi = {10.1214/ECP.v18-2551},
	journal = {Electronic Communications in Probability},
	author = {Meckes, Elizabeth and Meckes, Mark},
	month = jan,
	year = {2013},
	pages = {1--13}
}

@article{Girardi2025Dec,
	author = {Girardi, Filippo and Mele, Francesco Anna and Zhao, Haimeng and Fanizza, Marco and Lami, Ludovico},
	title = {{Random Stinespring superchannel: converting channel queries into dilation isometry queries}},
	journal = {arXiv},
	year = {2025},
	month = dec,
	date = {2025-12-23},
	urldate = {2025-12-24},
	eprint = {2512.20599},
	doi = {10.48550/arXiv.2512.20599}
}
\appendix

\section{A weaker lower bound using existing packing nets}\label{sec:LB-sqrt(eps)}

A natural approach to constructing such an ensemble in $\mathcal{E}(d_A, d_B, r, \eps, \eta)$ is to use packing nets. 
Assume that $d_B \geq 2$. From \cite[Lemma 14]{Girardi2025Dec}, we have
\bb
\log \cM\left(\mathcal{C}(d_A,d_B,\lfloor \tfrac{r}{2}\rfloor),\ \|\cdot\|_{\diamond},\ 1/2\right) = \Theta\left(r\,d_A d_B\right),
\ee
where $\cM(\mathcal{S}, |\cdot|, \delta)$ denotes the $\delta$-packing number of the set $\mathcal{S}$ with respect to the norm $|\cdot|$.

Let $M = \cM\big(\mathcal{C}(d_A,d_B,\lfloor \frac{r}{2}\rfloor),\ \|\cdot\|_{\diamond},\ 1/2\big)$, and let $\{\widetilde{\Phi}_x\}_{x\in[M]}$ be a $1/2$-diamond-norm packing of quantum channels, with corresponding Stinespring isometries $\{\widetilde{V}_x\}_{x\in[M]}$.

For a given $\varepsilon \in (0,\frac{1}{4})$ and each $x \in [M]$, we define the convex mixture
\bb\Phi_x = (1-4\eps)\widetilde{\Phi}_1 + 4\eps \widetilde{\Phi}_x.\ee
This is a valid quantum channel of Choi rank at most $2\lfloor \frac{r}{2}\rfloor \le r$. 
We observe that for any distinct $x,y \in [M]$,
\bb
\|\Phi_x - \Phi_y\|_\diamond
&= 4\varepsilon \|\widetilde{\Phi}_x - \widetilde{\Phi}_y\|_\diamond > 2\varepsilon,
\ee
since $\|\widetilde{\Phi}_x - \widetilde{\Phi}_y\|_\diamond > 1/2$ by the packing property.
Moreover, by \cite{kretschmann2008information}, we have 
\bb
\inf_{V_{\Phi_x}} \|V_{\Phi_x} - V_1\|_{\mathrm{op}}^2 \le \| \Phi_x -\Phi_1\|_{\diamond} =4\eps  \| \widetilde{\Phi}_x -\widetilde{\Phi}_1\|_{\diamond}\le 4\eps,
\ee 
where $V_1$ is a Stinespring isometry for $\Phi_1$. Let $V_x$ be a Stinespring isometry for $\Phi_x$ achieving this infimum. Then for all $x,y \in [M]$,
\bb
\|V_x - V_y\|_{\mathrm{op}}&\le \|V_x - V_1\|_{\mathrm{op}} +\|V_y - V_1\|_{\mathrm{op}} \le 4\sqrt{\eps}. 
\ee
Thus, $\{\Phi_x\}_{x\in[M]} \in \mathcal{E}(d_A,d_B, r,\varepsilon,4\sqrt{\varepsilon})$, which implies the lower bound from Theorem \ref{thm:gen-LB}
\begin{align}
    N \ge  \frac{(2/3)\log(M)-\log 2}{4\eta \log(d_Br/\eta)} \ge \Omega\left(\frac{d_Ad_Br}{\sqrt{\eps}\log(d_Br/\sqrt{\eps})}\right). 
\end{align}

\section{Existence of the quantum channel  $\widetilde{\Phi}_0$ }\label{sec:existence}
In this section, we want to show the existence of a quantum channel  $\widetilde{\Phi}_0$  with  Kraus operators $\{K_{0,i}\}_{i\in [r]}$ satisfying
\begin{align}
    \left|\Tr\left[ K_{0,i}^\dagger  K_{0,j}\right]\right| \le \frac{2d_A}{r}\delta_{i,j}\,, \quad \forall i, j\in [r]. 
\end{align} 
We make cases depending on whether $d_A \le d_B$ or not. 
\begin{itemize}
    \item \textbf{Case $1$:}  $d_A \le d_B$, let $k = \bigl\lfloor \frac{d_B}{d_A} \bigr\rfloor \ge 1$. 
We decompose $\mathbb{C}^{d_B} \simeq \bigl(\bigoplus_{i=1}^k \mathbb{C}^{d_A}\bigr) \oplus \mathbb{C}^{d_C}$, 
where $d_C = d_B - kd_A < d_A$.

For each block $A_i \simeq A$ ($i = 1, \dots, k$), we can choose $l = d_A^2$ orthogonal 
$d_A \times d_A$ unitary matrices $\{U_{i,j}\}_{j \in [l]}$ (for example  the 
generalised Pauli operators). Since $kl = \bigl\lfloor \frac{d_B}{d_A} \bigr\rfloor d_A^2 \ge \bigl\lceil\frac{d_Ad_B}{2}\bigl\rceil\ge\bigl\lceil \frac{r}{2}\bigr\rceil $, we may 
select a subset $S \subset [k] \times [l]$ with $|S| = \bigl\lceil \frac{r}{2}\bigr\rceil$. For each $(i,j) \in S$, define the 
Kraus operator
\bb
K_{i,j} = \left( 0 \oplus \sqrt{\tfrac{1}{|S|}}\, U_{i,j} \right),
\ee
where the direct sum is taken with respect to the decomposition above, and $U_{i,j}$ acts 
nontrivially only on the $i$-th $\mathbb{C}^{d_A}$ summand.

We then verify:
\begin{itemize}
    \item[(a)] \textbf{Completeness:}
    \bb
    \sum_{(i,j) \in S} K_{i,j}^\dagger K_{i,j}
    = \sum_{(i,j) \in S} \frac{1}{|S|} \, U_{i,j}^\dagger U_{i,j}
    = \mathbb{I}_A.
    \ee
    
    \item[(b)] \textbf{Orthogonality:} For all $(i,j), (i',j') \in S$,
    \bb
    \left|\Tr\!\big[ K_{i,j}^\dagger K_{i',j'} \big]\right|
    = \frac{d_A}{|S|} \, \delta_{i,i'} \delta_{j,j'} \le \frac{2d_A}{r} \, \delta_{i,i'} \delta_{j,j'}.
    \ee
    
    \item[(c)] \textbf{Kraus rank:} The number of Kraus operators is  $|S|\le r$.
\end{itemize}
    \item \textbf{Case $2$:} $d_A > d_B$, let $k = \bigl\lfloor \frac{d_A}{d_B} \bigr\rfloor \in [1, r]$ and write $d_A = k d_B + d_C$ with $0 \le d_C < d_B$. We can then decompose $\id_A = \id_{B_1} \oplus \cdots \oplus \id_{B_k} \oplus \id_{C}$, where each $B_i \simeq B$ (i.e., $\dim B_i = d_B$). 

    For each block $B_i$ ($i=1,\dots,k$), construct $l = \bigl\lceil \frac{r}{2k} \bigr\rceil \in [1, d_B^2]$ orthogonal $d_B \times d_B$ unitary matrices $\{U_{i,j}\}_{j \in [l]}$ that are supported on $B_i$ and define the corresponding $d_A \times d_B$ matrices
    \bb
    K_{i,j} = \bigl(0 \oplus \tfrac{1}{\sqrt{l}}\, U_{i,j}\bigr) \quad (j \in [l]),
    \ee
    where the direct sum is taken with respect to the decomposition $\mathbb{C}^{d_A} \simeq \bigl(\bigoplus_{i=1}^k \mathbb{C}^{d_B}\bigr) \oplus \mathbb{C}^{d_C}$ and $U_{i,j}$ acts nontrivially only on the $i$-th $d_B$-dimensional summand.
    
    For the remaining block $C$, since $d_C < d_B$ we can apply Case $1$ and construct  $r' =  \bigl\lceil \frac{r d_C}{2d_A} \bigr\rceil \in [1, d_C^2\lfloor\frac{d_B}{d_C}\rfloor]$ orthogonal $d_C \times d_B$ isometries $\{V_{i'}\}_{i' \in [r']}$  and define
    \bb
    K_{k+1, i'} = \bigl(0 \oplus \tfrac{1}{\sqrt{r'}}\, V_{i'}\bigr) \quad (i' \in [r']),
    \ee
    where now $V_{i'}$ acts nontrivially only on the $\mathbb{C}^{d_C}$ summand. We can check
    \begin{itemize}
        \item[(a)] \textbf{Completeness:}
        \bb
            \sum_{i=1}^k \sum_{j=1}^l K_{i,j}^\dagger K_{i,j} 
            + \sum_{i'=1}^{r'} K_{k+1,i'}^\dagger K_{k+1,i'} 
            &= \sum_{i=1}^k \sum_{j=1}^l \frac{1}{l}\,\mathbb{I}_{B_i} 
               + \sum_{i'=1}^{r'} \frac{1}{r'}\,\mathbb{I}_{C} \\
            &= \mathbb{I}_A.
        \ee
        
        \item[(b)] \textbf{Orthogonality:} For all $i, i' \in [k]$ and $j, j' \in [l]$,
        \bb
            \Tr\!\big[K_{i,j}^\dagger K_{i',j'}\big] 
            &= \delta_{i,i'}\delta_{j,j'} \frac{d_B}{l} 
            \le \delta_{i,i'}\delta_{j,j'} \frac{2d_A}{r},
        \ee
        and for $i', i'' \in [r']$,
        \bb
            \Tr\!\big[K_{k+1,i'}^\dagger K_{k+1,i''}\big] 
            = \delta_{i',i''} \frac{d_C}{r'} 
            \le \delta_{i',i''} \frac{2d_A}{r}.
        \ee
        
        \item[(c)] \textbf{Kraus rank:} The total number of Kraus operators is
        \bb
            lk + r' 
            = \Bigl\lfloor\frac{r}{2k}\Bigr\rfloor k 
            + \bigl\lceil \frac{r d_C}{2d_A} \bigr\rceil
            \le r.
        \ee
    \end{itemize}
\end{itemize}

\section{Proof of Lemma \ref{lem:ppts-Phi}}\label{app:first-moment-and-lip}
In this section we prove Lemma \ref{lem:ppts-Phi} which we restate.

\begin{lemma*}[(Restatement of Lemma \ref{lem:ppts-Phi})]
 Let us suppose that either $d_A = rd_B$ or $d_A\le rd_B/2$ and, given any unitary $U_i\in{\rm U}(rd_B)$, let $\Phi_i$ be the channel constructed from according to \eqref{construction-eq} or according to \eqref{eq:V} and \eqref{construction-leq}, respectively. Then, calling $J_{\Phi_i}$ the Choi state of the channel $\Phi_i$, the function 
 \bb
 f:(U_1, U_2)\in {\rm U}(rd_B)^2 \mapsto \|J_{\Phi_1} - J_{\Phi_2} \|_1
 \ee
is $L\coloneqq 4\sqrt{\frac{2}{d_A}}\eps$-Lipschitz with respect to the $\ell_2$-sum of the 2-norms, namely
\bb
    |f(U_1,U_2)-f(U'_1,U'_2)|\leq  L  \|(U_1,U_2) -(U_1',U_2')\|_2
\ee
for all $U_1,U_1',U_2,U_2'\in {\rm U}(rd_B)$, where $\|(A,B)\|_2\coloneqq\sqrt{\|A\|_2^2+\|B\|_2^2}$. 
Furthermore, if we consider independent random unitaries $U_1,U_2 \sim \Haar({\rm U}(rd_B))$, we have
\bb
\ex{\|J_{\Phi_1} - J_{\Phi_2}\|_1}\ge  \Omega(\eps). 
\ee
\end{lemma*}

\subsection{Proof that $f$ is Lipschitz in the case $d_A=rd_B$}
Let $\rho = \proj{\Psi}_{A'A}$.  We have that using $\|O\|_{\rm{op}} = \eps$,


\bb
\left\|\ptr{E}{(V_1 - V_1') \rho V_1^\dagger}  \right\|_{1} &\le \left\|{(V_1 - V_1') \rho V_1^\dagger}  \right\|_{1}
\\& = \left\|{( U_1 O U_1^\dagger - U_1' O U_1'^\dagger) \rho }  \right\|_{1}
\\&\le \left\|{( U_1-U_1') O U_1^\dagger \rho }  \right\|_{1} + \left\|{U_1'O ( U_1-U_1') ^\dagger \rho }  \right\|_{1}
\\&\le 2\eps \sqrt{\frac{1}{d_A}} \|{U}_1-{U}_1'\|_2.
\ee
Hence by the triangle inequality 
\bb
 \left\|\ptr{E}{V_1 \rho V_1^\dagger} -\ptr{E}{V_1' \rho V_1'^\dagger} \right\|_{1}
 &\le  \left\|\ptr{E}{(V_1 -V_1') \rho V_1^\dagger}\right\|_{1} +  \left\|\ptr{E}{V_1' \rho (V_1-V_1')^\dagger} \right\|_{1}
 \\&\le  4\eps \sqrt{\frac{1}{d_A}} \|{U}_1-{U}_1'\|_2.
\ee
Finally, again by the triangle inequality  we have that 
\bb
&|f(U_1, U_2) - f(U_1', U_2')| 
\\&= \left|\left\|\ptr{E}{V_1 \rho V_1^\dagger} -\ptr{E}{V_2 \rho V_2^\dagger} \right\|_{1} -\left\|\ptr{E}{V_1' \rho V_1'^\dagger} -\ptr{E}{V_2' \rho V_2'^\dagger} \right\|_{1} \right|
\\&\le  \left|\left\|\ptr{E}{V_1 \rho V_1^\dagger} -\ptr{E}{V_2 \rho V_2^\dagger} \right\|_{1} -\left\|\ptr{E}{V_1' \rho V_1'^\dagger} -\ptr{E}{V_2' \rho V_2'^\dagger} \right\|_{1} \right|
\\&\le   \left\|\ptr{E}{V_1 \rho V_1^\dagger} -\ptr{E}{V_1' \rho V_1'^\dagger} \right\|_{1} +\left\|\ptr{E}{V_2 \rho V_2^\dagger} -\ptr{E}{V_2' \rho V_2'^\dagger} \right\|_{1}
\\&\le  4\eps \sqrt{\frac{1}{d_A}} \|{U}_1-{U}_1'\|_2 +  4\eps \sqrt{\frac{1}{d_A}} \|{U}_2-{U}_2'\|_2
\\&\le  4 \eps\sqrt{\frac{2}{d_A}} \sqrt{\|{U}_1-{U}_1'\|_2^2+ \|{U}_2-{U}_2'\|_2^2},
\ee 
where we use $a+b \le \sqrt{2}\sqrt{a^2+b^2}$ in the last inequality.

\subsection{Proof that $f$ is Lipschitz in the case $d_A\le rd_B/2$}
We have that 
\bb
\left\|\ptr{E}{(V_1 - V_1') \rho V_1^\dagger}  \right\|_{1} &\le \left\|{(V_1 - V_1') \rho V_1^\dagger}  \right\|_{1}
\\& = \eps \left\|{( \ket{1}_F\otimes {U}_1 - \ket{1}_F\otimes {U}_1') S \rho }  \right\|_{1}
\\&\le \eps \sqrt{\bra{\Psi} S^\dagger ( \ket{1}_F\otimes {U}_1 - \ket{1}_F\otimes {U}_1')^\dagger ( \ket{1}_F\otimes {U}_1 - \ket{1}_F\otimes {U}_1') S \ket{\Psi}}
\\&= \eps \sqrt{\Tr[S^\dagger ( {U}_1 - {U}_1')^\dagger ( {U}_1 - {U}_1') S ]}
\\&\le \eps \sqrt{\frac{1}{d_A}} \|{U}_1-{U}_1'\|_2. 
\ee
We conclude the proof in this case following similar steps as in the previous case.

\subsection{Lower bound on the expected first moment in the case $d_A\le rd_B/2$}

We have that 
\bb
   \left\|J_{{\Phi}_1} -J_{{\Phi}_2}\right\|_{1}&=   \left\|{\Phi}_1(\Psi) -{\Phi}_2(\Psi)\right\|_1\\
    &=  \left\|\ptr{E}{{V}_1(\Psi){V}_1^\dagger} -\ptr{E}{{V}_2(\Psi){V}_2^\dagger}\right\|_1
    \\&= \bigg\|\eps^2 \proj{1} \otimes \left(\ptr{E}{\widetilde{V}_1\Psi \widetilde{V}_1^\dagger}-\ptr{E}{\widetilde{V}_2\Psi \widetilde{V}_2^\dagger}\right)\\ 
    &\qquad + \eps\sqrt{1-\eps^2} \ket{0}\bra{1}\otimes \ptr{E}{\widetilde{V}_0\Psi ( \widetilde{V}_1^\dagger - \widetilde{V}_2^\dagger)}  \\
    &\qquad + \eps\sqrt{1-\eps^2} \ket{1}\bra{0}\otimes \ptr{E}{( \widetilde{V}_1 - \widetilde{V}_2)\Psi \widetilde{V}_0^\dagger}  \bigg\|_1
    \\&\geqt{(i)}   \eps\sqrt{1-\eps^2} \left\|\ket{0}\bra{1}\otimes \ptr{E}{\widetilde{V}_0\Psi (\widetilde{V}_1^\dagger - \widetilde{V}_2^\dagger)}  + \ket{1}\bra{0}\otimes \ptr{E}{( \widetilde{V}_1 - \widetilde{V}_2)\Psi \widetilde{V}_0^\dagger}  \right\|_1\\
    &\quad-  \eps^2\left\| \proj{1} \otimes (\widetilde{V}_1\Psi \widetilde{V}_1^\dagger-\widetilde{V}_2\Psi \widetilde{V}_2^\dagger)   \right\|_1
      \\&\eqt{(ii)}   2\eps\sqrt{1-\eps^2} \left\| \ptr{E}{\widetilde{V}_0\Psi ( \widetilde{V}_1^\dagger - \widetilde{V}_2^\dagger)}  \right\|_1 - 2 \eps^2
\ee
where in (i) we have used the reverse triangle inequality and the bound
\bb\label{eq:DPI}
    \big\|\ptr{E}{X_{FBE}}\big\|_1\leq \big\|X_{FBE}\big\|_1,
\ee
which holds for every operator $X_{FBE}$ and follows from the data-processing inequality for the trace norm; in (ii) we have noticed that
\bb
    &\big\|\ket{0}\bra{1}\otimes X_{BE}+\ket{1}\bra{0}\otimes X_{BE}^\dagger\big\|_1\\
    &\quad =\Tr\sqrt{\left(\ket{0}\bra{1}\otimes X_{BE}+\ket{1}\bra{0}\otimes X_{BE}^\dagger\right)^\dagger \left(\ket{0}\bra{1}\otimes X_{BE}+\ket{1}\bra{0}\otimes X_{BE}^\dagger\right)}\\
    &\quad =\Tr\sqrt{\ketbra{1}\otimes X_{BE}^\dagger X_{BE}+\ketbra{0}\otimes X_{BE} X_{BE}^\dagger}\\
    &\quad = \Tr\sqrt{X_{BE}^\dagger X_{BE}}+\Tr\sqrt{X_{BE} X_{BE}^\dagger}\\
    &\quad =2\|X_{BE}\|_1;
\ee
and  we have upper bounded $\|\widetilde{V}_1\Psi \widetilde{V}_1^\dagger-\widetilde{V}_2\Psi \widetilde{V}_2^\dagger\|_1\leq \|\widetilde{V}_1\Psi \widetilde{V}_1^\dagger\|_1+\|\widetilde{V}_2\Psi \widetilde{V}_2^\dagger\|_1=2$. 

Let us define the operator $C \coloneqq  \ptr{E}{\widetilde{V}_0\Psi_{A'A} ( \widetilde{V}_1^\dagger - \widetilde{V}_2^\dagger)}$. We want to prove that
    \begin{align}
        \mathbb{E}\Tr\big[|C|^2\big]&= \frac{2}{r}.
    \end{align}
 Let $\{\ket{i}_{E}\}_{i\in[r]}$ be an orthonormal basis for $E$. For $i=1,\dots, r$, let $K_{0, i}^{A\to B} \coloneqq \bra{i}_E \widetilde{V}_0^{A\to BE} $ and $K_{1, i}^{A\to B} \coloneqq \bra{i}_E \widetilde{V}_1^{A\to BE}$ be the Kraus operators obtained from the isometries $\widetilde{V}_0$ and $\widetilde{V}_1$, respectively. Writing the trace on the system $E$ in terms of the basis $\{\ket{i}_E\}_{i\in[r]}$, we get 
\bb\label{eq:second}
   \mathbb{E}\Tr\big[|C|^2\big]&= \mathbb{E}\Tr\left[\sum_{i,j=1}^r K_{0, i} \Psi_{A'A} ( K_{1, i} - K_{2, i})^\dagger ( K_{1, j} - K_{2, j}) \Psi_{A'A} K_{0, j}^\dagger \right]\\
   &\eqt{(iii)}  \sum_{i,j=1}^r \Tr\left[ K_{0, i} \Psi_{A'A} \left(\frac{2\id}{r}\delta_{i,j}\right) \Psi_{A'A} K_{0, j}^\dagger \right]
    \\&=  \frac{2}{r}\sum_{i=1}^r \Tr\left[ K_{0, i} \Psi_{A'A}^2 K_{0, i}^\dagger \right]\\
    &=  \frac{2}{r} \Tr\left[  \Psi_{A'A} \sum_{i=1}^r K_{0, i}^\dagger K_{0, i} \right]\\
    &=\frac{2}{r},
\ee
where in (iii) we have expanded
\bb
    \ex{(K_{1, i} - K_{2, i})^\dagger ( K_{1, j} - K_{2, j})}&=
    \ex{K_{1, i}^\dagger K_{1, j}}+\ex{K_{2, i}^\dagger K_{2, j}}\\
    &\quad -\ex{K_{1, i}^\dagger K_{2, j}}-\ex{K_{2, i}^\dagger K_{1, j}}
\ee
and, for $z_1,z_2\in\{1, 2\}$, we have computed
\bb\label{eq:54}
    \ex{K_{{z_1}, i}^\dagger K_{{z_2}, j}} &= \ex{ \widetilde{V}_{z_1}^\dagger  \ket{i}_E\bra{j}_E \widetilde{V}_{z_2}  }
    = S^\dagger\ex{   U_{z_1}^\dagger \ket{i}_E\bra{j}_E U_{z_2}  }S
    \\&=  S^\dagger \left(\frac{\delta_{{z_1},{z_2}}}{rd_B}  \Tr\left[\ket{i}_E\bra{j}_E\otimes \id_{B}\right]\right)S
    = \frac{\delta_{{z_1},{z_2}}\delta_{i,j}}{r} \id_{A},
\ee
leveraging the fact that $\EE{U\in{\rm U}(d)}[U]=0$, $\EE{U\in{\rm U}(d)}[U^\dagger X U]=\frac{\Tr[X]}{d}\id $ and $S^\dagger S=\id_A$. \\

The only inequality we are left to prove is
\bb
   \mathbb{E}\Tr\big[|C|^4\big] \le  \frac{128}{r^3}.
\ee
We have
\bb\label{eq:ABCD-initial}
   \mathbb{E}\Tr\big[|C|^4\big]&= \mathbb{E}\Tr\Bigg[ \sum_{i,j=1}^r K_{0, i} \Psi_{A'A} ( K_{1, i} - K_{2, i})^\dagger ( K_{1, j} - K_{2, j}) \Psi_{A'A} K_{0, j}^\dagger  \\[-1em]
   &\qquad \qquad \qquad\qquad  \times \sum_{k,l=1}^r K_{1, k}\Psi_{A'A} ( K_{1, k} - K_{2, k})^\dagger ( K_{1, l} - K_{2, l}) \Psi_{A'A} K_{0, l}^\dagger\Bigg]
   \\&= \mathbb{E} \sum_{i,j=1}^r \sum_{k,l=1}^r \bra{\Psi_{A'A}}K_{0, l}^\dagger K_{0, i}\ket{\Psi_{A'A}}\bra{\Psi_{A'A}}K_{0, j}^\dagger K_{0, k} \ket{\Psi_{A'A}}     \\[-1.2em]
   &\phantom{\mathbb{E} \sum_{i,j=1}^r \sum_{k,l=1}^r }\quad \times  \bra{\Psi_{A'A}} ( K_{1, i} - K_{2, i})^\dagger ( K_{1, j} - K_{2, j}) \ket{\Psi_{A'A}}  \\[-1.2em]
   &\phantom{\mathbb{E} \sum_{i,j=1}^r \sum_{k,l=1}^r }\quad \times \bra{\Psi_{A'A}} ( K_{1, k} - K_{2, k})^\dagger ( K_{1, l} - K_{2, l}) \ket{\Psi_{A'A}}
    \\&\leqt{(iv)} \mathbb{E} \sum_{i,j=1}^r  \frac{4}{r^2} \left|\bra{\Psi_{A'A}} ( K_{1, i} - K_{2, i})^\dagger ( K_{1, j} - K_{2, j}) \ket{\Psi_{A'A}} \right|^2
\ee
where in (iv) we have noticed that, by \eqref{eq:Kraus-1},
\bb
    &\big|\bra{\Psi_{A'A}}K_{0, l}^\dagger K_{0, i}\ket{\Psi_{A'A}}\bra{\Psi_{A'A}}K_{0, j}^\dagger K_{0, k} \ket{\Psi_{A'A}}\big|
    \\&\qquad \qquad = \left|\frac{1}{d_A} \Tr[K_{0, l}^\dagger K_{0, i}]\frac{1}{d_A} \Tr[K_{0, j}^\dagger K_{0, k} ]\right| \leq \frac{4}{r^2} \delta_{i,l}\delta_{j,k}.
\ee
Hence 
\bb\label{eq:ABCD}
   \mathbb{E}\Tr\big[|C|^4\big]&\le\mathbb{E} \sum_{i,j=1}^r  \frac{4}{r^2} \left|\bra{\Psi_{A'A}} ( K_{1, i} - K_{2, i})^\dagger ( K_{1, j} - K_{2, j}) \ket{\Psi_{A'A}} \right|^2
    \\&\leqt{(v)} \mathbb{E} \sum_{i,j=1}^r  \frac{16}{r^2} \Big(\left|\bra{\Psi_{A'A}} K_{1, i} ^\dagger  K_{1, j} \ket{\Psi_{A'A}} \right|^2 +\left|\bra{\Psi_{A'A}} K_{2, i} ^\dagger  K_{2, j} \ket{\Psi_{A'A}} \right|^2 \\[-1em]
    & \phantom{\le \mathbb{E} \sum_{i,j=1}^r  \frac{16}{r^2}}\quad +2 \left|\bra{\Psi_{A'A}} K_{1, i} ^\dagger  K_{2, j} \ket{\Psi_{A'A}} \right|^2\Big)
    \\&\leqt{(vi)}  \mathbb{E} \sum_{i,j=1}^r  \frac{32}{r^2} \left(\left|\bra{\Psi_{A'A}} K_{1, i} ^\dagger  K_{1, j} \ket{\Psi_{A'A}} \right|^2 +\left|\bra{\Psi_{A'A}} K_{2, i} ^\dagger  K_{2, j} \ket{\Psi_{A'A}} \right|^2 \right)
     \\& =   \frac{64}{r^2} \sum_{i,j=1}^r  \mathbb{E} \left|\bra{\Psi_{A'A}} K_{1, i} ^\dagger  K_{1, j} \ket{\Psi_{A'A}} \right|^2,
\ee
where in (v) and in (vi) we have leveraged the inequality $2ab\leq a^2+b^2$ multiple times.

Recalling that we defined $\widetilde{V}_1 = U_1S$ and $K_{1, i} = \bra{i}_E \widetilde{V}_1$ , we compute
\bb\label{eq:73}
&\mathbb{E} \left|\bra{\Psi_{A'A}} K_{1, i} ^\dagger  K_{1, j} \ket{\Psi_{A'A}} \right|^2\\
&\quad = \frac{1}{d_A^2}\ex{\Tr[ K_{1, i} ^\dagger  K_{1, j} ]\Tr[ K_{x,j} ^\dagger  K_{1, i} ]}
\\&\quad = \frac{1}{d_A^2} \sum_{k,l =1}^{d_A} \mathbb{E}\Tr\left[  K_{1, i} ^\dagger K_{1, j}  \ket{l}\bra{k}  K_{1, j}^\dagger  K_{1, i}  \ket{k}\bra{l} \right]
\\&\quad \eqt{(vii)} \frac{1}{d_A^2} \sum_{k,l =1}^{d_A} \mathbb{E}\Tr\left[   U_1^\dagger  \big(\ket{i}_E\bra{j}_E\otimes \id_B\big) U_1 S  \ket{l}\bra{k}   S^\dagger U_1^\dagger  \big(\ket{j}_E\bra{i}_E\otimes \id_B \big)U_1 S \ket{k}\bra{l}  S^\dagger  \right]
\\&\quad \eqt{(viii)}\frac{1}{d_A^2} \sum_{k,l =1}^{d_A}  \sum_{\alpha, \beta \in S_2} \operatorname{Wg}(\beta \alpha, d_Br) \ptr{ \beta}{ S  \ket{k}\bra{l}  S^\dagger, S \ket{l}\bra{k}   S^\dagger}\\[-1em]
    &\quad \phantom{ \eqt{(xii)} \sum_{\alpha, \beta \in S_2} \operatorname{Wg}(\beta \alpha, d_Br)}\quad \times \ptr{ \alpha (12)}{\ket{i}_E\bra{j}_E\otimes \id_B, \ket{j}_E\bra{i}_E\otimes \id_B} \\
    &\quad = \frac{1}{d_A^2} \sum_{k,l =1}^{d_A}\Big( \operatorname{Wg}((1)(2), d_Br) \big(\delta_{k,l}d_B+\delta_{i,j}d_B^2\big)+\operatorname{Wg}((12), d_Br) \big(\delta_{k,l}\delta_{i,j}d_B^2+d_B\big)\Big)\\
    &\quad\eqt{(ix)} \frac{1}{d_A^2}\cdot  \frac{1}{(d_Br)^2-1}\sum_{k,l =1}^{d_A} \Big(\delta_{k,l}d_B+\delta_{i,j}d_B^2-\frac{1}{d_Br} \big(\delta_{k,l}\delta_{i,j}d_B^2+d_B\big)\Big)\\
    &\quad =  \frac{1}{d_A} \cdot \frac{1}{(rd_B)^2-1}\Big(d_B+\delta_{i,j}d_Ad_B^2-\frac{1}{d_Br} \big(\delta_{i,j}d_B^2+d_Ad_B\big)\Big)
\ee
where in (vii) have expanded $K_{x,i}=\bra{i}_EU_1S$ and we have leveraged the ciclicity of the trace; in (viii) we have used  Lemma \ref{lem:Wg} with $A_1=\ket{i}_E\bra{j}_E\otimes \id_B$, $A_2=\ket{j}_E\bra{i}_E\otimes \id_B$, $B_1=S\ket{k}\bra{l}S^\dagger$ and $B_2=S\ket{l}\bra{k}S^\dagger$; in (ix) we have used the values given in Lemma \ref{lem:wg2}. 
Combining \eqref{eq:ABCD} with \eqref{eq:73}, we get

\bb\label{eq:fourth}
    \mathbb{E}\Tr\big[|C|^4\big]&\leq \frac{64}{r^2} \sum_{i,j=1}^r  \mathbb{E} \left|\bra{\Psi_{A'A}} K_{1, i} ^\dagger  K_{1, j} \ket{\Psi_{A'A}} \right|^2\\
    &\leq \frac{64}{r^2}\cdot  \frac{1}{d_A}\cdot\frac{1}{(rd_B)^2-1} \sum_{i,j=1}^r \Big(d_B+\delta_{i,j}d_Ad_B^2 -\frac{1}{d_Br} \big(\delta_{i,j}d_B^2+d_Ad_B\big)\Big)\\
    &=\frac{64}{r^2}\cdot \frac{1}{d_Ad_B}\cdot\frac{1}{(rd_B)^2-1}\left( d_B^2r^2 + d_Ad_Br^3 - d_B^2 - d_Ad_Br \right)\\
     &=\frac{64}{r^2}\cdot \left( \frac{1}{d_Ad_B} + \frac{1}{r}  + \frac{1-d_B^2+ d_Ad_B-d_Ad_Br}{d_Ad_B(r^2d_B^2-1)}\right)\leq \frac{128}{r^3}, 
\ee
where in the last line we have recalled that $r\leq d_Ad_B$.

 By Hölder's inequality applied to  $\mathbb{E}\Tr[\,\cdot\,]$ with conjugate exponents $3$ and $3/2$, we get

\bb
    \mathbb{E}[\Tr\big[|C|^2\big]]=\mathbb{E}[\Tr\big[|C|^{4/3} |C|^{2/3}\big]]
    \le \left( \mathbb{E}\left[ \Tr\big[|C|^4\big]\right] \right)^{1/3} \left(\mathbb{E}\left[\Tr\big[|C|\big]\right]\right)^{2/3},
\ee
which yields
\bb\label{eq:Hölder}
    \big(\mathbb{E}[\Tr\big[|C|^2\big]]\big)^3
    \le   \mathbb{E}\left[ \Tr\big[|C|^4\big]\right] \left(\mathbb{E}\left[\Tr\big[|C|\big]\right]\right)^{2}.
\ee
Whence, by  \eqref{eq:second} and \eqref{eq:fourth} combined with \eqref{eq:Hölder}, we have
\begin{align}\label{eq:LB-1st moment}
    \big(\mathbb{E}\Tr\big[|C|\big]\big)^2 \ge \frac{\big(\mathbb{E}\Tr\big[|C|^2\big]\big)^3}{\mathbb{E}\Tr\big[|C|^4\big]}\ge \frac{(\frac{2}{r})^3}{\frac{128}{r^3}} = \frac{1}{16}. 
\end{align}
Therefore 
\bb
   \ex{\left\|J_{{\Phi}_1} -J_{{\Phi}_2}\right\|_{1}}& \ge    2\eps\sqrt{1-\eps^2} \ex{ \left\| \ptr{E}{\widetilde{V}_0\Psi ( \widetilde{V}_1^\dagger - \widetilde{V}_2^\dagger)}  \right\|_1} - 2 \eps^2
     \\&\ge  0.5\eps\sqrt{1-\eps^2} - 2\eps^2 
     \\&\ge\Omega(\eps). 
\ee

\subsection{Lower bound on the expected first moment in the case $d_A = rd_B$}
Recall the definition of the $d_A\times d_A$ matrix
\bb
    O\coloneqq\begin{cases}
        \eps\cdot {\rm diag}( \mathrm{e}^{i\theta}, \mathrm{e}^{-i\theta},\dots,  \mathrm{e}^{i\theta}, \mathrm{e}^{-i\theta}) & \text{if $d_A$ is even,}\\
        \eps\cdot {\rm diag}( \mathrm{e}^{i\theta}, \mathrm{e}^{-i\theta},\dots,  \mathrm{e}^{i\theta}, \mathrm{e}^{-i\theta},0) & \text{if $d_A$ is odd.}
    \end{cases}
\ee
where $\theta\in(\pi/2,\pi]$ satisfying $\epsilon =-2\cos\theta$. 
A  corresponding traceless matrix $\bar O$ is
\bb
    \bar{O}&\coloneqq O - \Tr[O]\frac{\id}{d_A} =\begin{cases}
        \epsilon\cdot {\rm diag}( \mathrm{e}^{i\theta}, \mathrm{e}^{-i\theta},\dots,  \mathrm{e}^{i\theta}, \mathrm{e}^{-i\theta})+\frac{\epsilon^2}{2}\id_A & \text{if $d_A$ is even,}\\
        \epsilon\cdot {\rm diag}( \mathrm{e}^{i\theta}, \mathrm{e}^{-i\theta},\dots,  \mathrm{e}^{i\theta}, \mathrm{e}^{-i\theta},0)+\frac{\epsilon^2}{2}\frac{d_A-1}{d_A}\id_A & \text{if $d_A$ is odd,}
    \end{cases}
\ee
where we have noticed that $\Tr \bar{O} =\epsilon \,\lfloor d_A/2\rfloor\,2\cos\theta =-\lfloor d_A/2\rfloor\epsilon^2$. This construction ensures the following properties, which will be used later:
\begin{enumerate}[label=(\alph*)]
    \item $(O+O^\dagger + O^\dagger O)_j= \epsilon \big((\mathrm{e}^{i\theta}+\mathrm{e}^{-i\theta})-2\cos\theta\big) =0 $ for all $j=1, \dots, d_A$,
    \item $\|\bar{O}\|_{\rm op}=\Theta(\eps )$,
    \item $\Tr[\bar{O} \bar{O}^\dagger]= d_A(\epsilon^2+\Theta(\eps^4))$,
    \item $\Tr[\bar{O}\bar{O}^\dagger\bar{O}\bar{O}^\dagger] = d_A (\eps^4 + \Theta(\eps^6))$.
\end{enumerate}
 Let $U\in{\rm U}(d_Br)$, we can define the Stinespring isometry
\bb\label{eq:Stinespring}
V \coloneqq U(\id + O)U^\dagger.
\ee
Note that $V$ is indeed an isometry, since, using property (a), we have
\bb
    V^\dagger V&= U(\id + O^\dagger)(\id + O)U^\dagger=\id +U(O+O^\dagger+O^\dagger O)U^\dagger=\id.\\
\ee
Furthermore, if we consider two unitaries $U_1$ and $U_2$, and we define the corresponding Stinespring isometries $V_1$ and $V_2$ as in \eqref{eq:Stinespring},  and let 
\bb
A_1 &\coloneqq \ptr{E}{U_1\bar{O}U_1^\dagger \Psi},\qquad
D(U_1,U_2)\coloneqq A_1 + A_1^\dagger -A_2 -A_2^\dagger.
\ee
We have that 
\bb
   \left\|J_{{\Phi}_1} -J_{{\Phi}_1}\right\|_{1}&=   \left\|{\Phi}_1(\Psi) -{\Phi}_2(\Psi)\right\|_1\\
    &\geq  \left\|
    \ptr{E}{U_1OU_1^\dagger \Psi + \Psi U_1O^\dagger U_1^\dagger -U_2OU_2^\dagger \Psi - \Psi U_2O^\dagger U_2^\dagger}
    \right\|_1\\
    &\quad -\left\|
    \ptr{E}{U_1OU_1^\dagger \Psi  U_1O^\dagger U_1^\dagger -U_2OU_2^\dagger \Psi  U_2O^\dagger U_2^\dagger}
    \right\|_1
    \\&\geq \| D(U_1,U_2)\|_1 -2\eps^2,
\ee
using that $ \left\|
    \ptr{E}{U_1OU_1^\dagger \Psi  U_1O^\dagger U_1^\dagger}\right\|_1 =\Tr  \left[
    {U_1OU_1^\dagger \Psi  U_1O^\dagger U_1^\dagger}\right] = \frac{1}{d_A} \Tr\left[OO^\dagger\right]\le \eps^2$.

Observe that
\bb
& \ptr{E}{U_1OU_1^\dagger \Psi + \Psi U_1O^\dagger U_1^\dagger -U_2OU_2^\dagger \Psi - \Psi U_2O^\dagger U_2^\dagger}
\\&= \ptr{E}{U_1\bar{O}U_1^\dagger \Psi + \Psi U_1\bar{O}^\dagger U_1^\dagger -U_2\bar{O}U_2^\dagger \Psi - \Psi U_2\bar{O}^\dagger U_2^\dagger}
\\&= D(U_1,U_2).
 \ee

  Now, we are interested in proving the bound
    \bb\label{eq:second-D}
    \ex{\Tr[|D|^2] } \geq \frac{2}{r}\cdot  \frac{\Tr\left[  \bar{O}^\dagger  \bar{O} \right]}{d_A}.
    \ee
    Since 
    \bb
    \ex{\Tr[|D|^2]} = 2\ex{\Tr[A_1A_1]+ 2\Tr[A_1^\dagger A_1] + \Tr[A_1^\dagger A_1^\dagger ] } -  2\Re \ex{\Tr[A_1A_2]+ \Tr[A_1^\dagger A_2]}.
    \ee
    Since $\Tr[\bar{O}] = 0$ we have that $\ex{U\bar{O}U^\dagger} = \Tr[\bar{O}] \frac{\dI}{d_A} = 0$, hence 
    \bb
    \ex{\Tr[A_1A_2]} =\ex{ \Tr[A_1^\dagger A_2]} = 0.
    \ee
    We have that 
    \bb
    \ex{\Tr[A_1A_1^\dagger ]} &= \ex{\Tr\left[\ptr{E}{U_1\bar{O}U_1^\dagger \Psi} \ptr{E}{\Psi U_1\bar{O}^\dagger U_1^\dagger}\right]} \\&=\sum_{i,j=1}^r  \ex{\Tr\left[{U_1\bar{O}U_1^\dagger \Psi} \ket{i}\bra{j} \otimes \id_B  {\Psi U_1\bar{O}^\dagger U_1^\dagger} \ket{j}\bra{i} \otimes \id_B\right]} 
    \\&=\sum_{i,j=1}^r  \bra{\Psi}  \ket{i}\bra{j} \otimes \id_B \ket{\Psi} \ex{\Tr\left[U_1\bar{O}U_1^\dagger \Psi U_1\bar{O}^\dagger U_1^\dagger \ket{j}\bra{i} \otimes \id_B\right]} 
    \\&=\sum_{i,j=1}^r \frac{\delta_{i,j} d_B}{d_A}\ex{\Tr\left[U_1\bar{O}U_1^\dagger \Psi U_1\bar{O}^\dagger U_1^\dagger \ket{j}\bra{i} \otimes \id_B\right]} 
    \\&=\sum_{i}^r \frac{ d_B}{d_A}\ex{\Tr\left[U_1\bar{O}U_1^\dagger \Psi U_1\bar{O}^\dagger U_1^\dagger \ket{i}\bra{i} \otimes \id_B\right]} 
    \\&= \frac{ d_B}{d_A}\ex{\Tr\left[U_1\bar{O}U_1^\dagger \Psi U_1\bar{O}^\dagger U_1^\dagger \right]} 
    \\&= \frac{ d_B}{d_A}\ex{\Tr\left[\bar{O}U_1^\dagger \Psi U_1 \bar{O}^\dagger  \right]} 
    \\&= \frac{ d_B}{d_A} \cdot \frac{1}{d_A} \Tr\left[ U_1 \bar{O}^\dagger  \bar{O}U_1^\dagger   \right]
    \\&= \frac{ 1}{r} \cdot \frac{1}{d_A} \Tr\left[  \bar{O}^\dagger  \bar{O} \right].
    \ee 
Using Lemma \ref{lem:Wg} and Lemma \ref{lem:wg2}, we can estimate  
\bb
    &\ex{\Tr[A_1A_1 ]} \\
    &\quad = \ex{\Tr\left[\ptr{E}{U_1\bar{O}U_1^\dagger \Psi} \ptr{E}{U_1\bar{O} U_1^\dagger \Psi }\right]} \\&\quad=\sum_{i,j=1}^r  \ex{\Tr\left[{U_1\bar{O}U_1^\dagger \Psi} \ket{i}\bra{j} \otimes \id_B  { U_1\bar{O} U_1^\dagger \Psi} \ket{j}\bra{i} \otimes \id_B\right]}
    \\&\quad= \sum_{i,j=1}^r \frac{1}{d_A^2} \ex{\Tr\left[\ket{i}\bra{j} \otimes \id_B   U_1\bar{O} U_1^\dagger\right]\Tr\left[\ket{j}\bra{i} \otimes \id_B U_1\bar{O}U_1^\dagger\right]}
    \\&\quad= \sum_{i,j=1}^r \frac{1}{d_A^2}  \sum_{a,b=1}^{d_A}\ex{\Tr\left[\ket{i}\bra{j} \otimes \id_B   U_1\bar{O} U_1^\dagger\ket{a}\bra{b}\ket{j}\bra{i} \otimes \id_B U_1\bar{O}U_1^\dagger\ket{b}\bra{a}\right]}
    \\&\quad=  \sum_{i,j=1}^r \frac{1}{d_A^2}  \sum_{a,b=1}^{d_A}  \frac{1}{d_A^2-1}\left( \Tr[\bar{O}\bar{O}^\dagger ] \Tr[\ket{a}\bra{b}\ket{j}\bra{i} \otimes \id_B] \Tr[\ket{b}\bra{a}\ket{i}\bra{j} \otimes \id_B] \right)
    \\&\quad\qquad -\sum_{i,j=1}^r \frac{1}{d_A^2}  \sum_{a,b=1}^{d_A}  \frac{1}{d_A(d_A^2-1)}\left( \Tr[\bar{O}\bar{O}^\dagger ] \Tr[\ket{a}\bra{b}\ket{j}\bra{i} \otimes \id_B\ket{b}\bra{a}\ket{j}\bra{i} \otimes \id_B] \right)
    \\&\quad=  \sum_{i,j=1}^r \frac{1}{d_A^2}   \frac{1}{d_A^2-1}\left( \Tr[\bar{O}\bar{O}^\dagger ] \Tr[\ket{j}\bra{i} \otimes \id_B \ket{i}\bra{j} \otimes \id_B\right)
    \\&\quad\qquad -\sum_{i,j=1}^r \frac{1}{d_A^2}  \sum_{a,b=1}^{d_A}  \frac{1}{d_A(d_A^2-1)}\left( \Tr[\bar{O}\bar{O}^\dagger ] \delta_{i,j}d_B^2 \right)
    \\&\quad= \frac{1}{d_A^2(d_A^2-1)}(d_B r^2 - r)\Tr[\bar{O}\bar{O}^\dagger ]
    \\&\quad= \frac{r}{d_A^2(d_A+1)}\Tr[\bar{O}\bar{O}^\dagger ].
\ee 
Therefore,
 \bb
    \ex{\Tr[|D|^2]} &= 2\ex{\Tr[A_1A_1]+ 2\Tr[A_1^\dagger A_1] + \Tr[A_1^\dagger A_1^\dagger ] } -  2\Re \ex{\Tr[A_1A_2]+ \Tr[A_1^\dagger A_2]}
    \\& = 2\frac{1}{r} \frac{\Tr\left[  \bar{O}^\dagger  \bar{O} \right]}{d_A} + 4\frac{r}{d_A^2(d_A+1)}\Tr[\bar{O}\bar{O}^\dagger ]
    \\&\ge \frac{2}{r} \frac{\Tr\left[  \bar{O}^\dagger  \bar{O} \right]}{d_A}.
    \ee
Finally, let us prove that
    \bb\label{eq:fourth-D}
    \ex{\Tr[|D|^4] } \le 4^4 \frac{2}{r^2d_A^3}\left( (\Tr[\bar{O}\bar{O}^\dagger])^2 d_B + \Tr[\bar{O}\bar{O}^\dagger\bar{O}\bar{O}^\dagger]r\right).
    \ee
    Recall that $D(U_1,U_2) = A_1 + A_1^\dagger -A_2 -A_2^\dagger$ so by the triangle inequality and Hölder inequality:
    \bb
     \ex{\Tr[|D|^4] } &= \ex{\|A_1 + A_1^\dagger -A_2 -A_2^\dagger\|_4^4}
     \\&\le \ex{ (\|A_1\|_4 + \|A_1^\dagger\|_4 + \|A_2\|_4 +\|A_2^\dagger\|_4)^4}
     \\&\le 4^3 \ex{ \|A_1\|_4^4 + \|A_1^\dagger\|_4^4 + \|A_2\|_4^4 +\|A_2^\dagger\|_4^4}
     \\&= 4^4 \ex{ \|A_1\|_4^4}.
    \ee
   Moreover,
   \bb
&\ex{ \|A_1\|_4^4} 
\\&= \ex{\Tr[A_1 A_1^\dagger A_1 A_1^\dagger ]}
\\&= \ex{\Tr\left[\ptr{E}{U_1\bar{O}U_1^\dagger \Psi} \ptr{E}{\Psi \bar{O}^\dagger  U_1^\dagger }\ptr{E}{U_1\bar{O}U_1^\dagger \Psi} \ptr{E}{\Psi U_1\bar{O}^\dagger  U_1^\dagger  }\right]} 
\\&=\sum_{i,j,k,l=1}^r  \mathbb{E}\Tr\Big[{U_1\bar{O}U_1^\dagger \Psi} (\ket{i}\bra{j} \otimes \id_B ) { \Psi U_1\bar{O}^\dagger  U_1^\dagger } (\ket{j}\bra{k} \otimes \id_B)\\[-1em]
&\phantom{=\sum_{i,j,k,l=1}^r  \mathbb{E}\Tr\Big[}\quad
\times  U_1\bar{O}U_1^\dagger \Psi (\ket{k}\bra{l} \otimes \id_B)  { \Psi U_1\bar{O}^\dagger  U_1^\dagger }( \ket{l}\bra{i} \otimes \id_B)\Big]
\\&=\sum_{i,j,k,l=1}^r  \frac{\delta_{i,j}d_B}{d_A}\cdot \frac{\delta_{k,l}d_B}{d_A}\mathbb{E}\Tr\Big[U_1\bar{O}U_1^\dagger  \Psi U_1\bar{O}^\dagger U_1^\dagger  (\ket{j}\bra{k} \otimes \id_B)\\[-1em]
&\phantom{=\sum_{i,j,k,l=1}^r  \frac{\delta_{i,j}d_B}{d_A}\cdot \frac{\delta_{k,l}d_B}{d_A}\mathbb{E}\Tr\Big[}
 \quad \times U_1\bar{O}U_1^\dagger     \Psi U_1\bar{O}^\dagger U_1^\dagger  (\ket{l}\bra{i} \otimes \id_B)\Big]
\\&=\sum_{i,k=1}^r  \frac{d_B^2}{d_A^2}\cdot \frac{1}{d_A^2} \ex{\Tr\left[U_1\bar{O}^\dagger U_1^\dagger (\ket{i}\bra{k} \otimes \id_B)U_1\bar{O}U_1^\dagger \right]\Tr\left[U_1\bar{O}^\dagger U_1^\dagger  (\ket{k}\bra{i} \otimes \id_B)U_1\bar{O}U_1^\dagger\right]}
\\&=\sum_{i,k=1}^r  \frac{d_B^2}{d_A^2}\cdot \frac{1}{d_A^2} \sum_{a,b=1}^{d_A} \ex{\Tr\left[U_1\bar{O}\bar{O}^\dagger U_1^\dagger  (\ket{i}\bra{k} \otimes \id_B) \ket{a}\bra{b} U_1\bar{O}\bar{O}^\dagger U_1^\dagger  (\ket{k}\bra{i} \otimes \id_B          \ket{b}\bra{a}\right]}
\\&=\sum_{i,k=1}^r  \frac{d_B^2}{d_A^4} \frac{1}{d_A^2-1}\left( (\Tr[\bar{O}\bar{O}^\dagger])^2 \delta_{k,i}d_B^2 + \Tr[\bar{O}\bar{O}^\dagger\bar{O}\bar{O}^\dagger] d_B\right) 
\\&\qquad -\frac{1}{d_A(d_A^2-1)}\left( (\Tr[\bar{O}\bar{O}^\dagger])^2 d_B + \Tr[\bar{O}\bar{O}^\dagger\bar{O}\bar{O}^\dagger] \delta_{k,i}d_B^2\right)
\\&\le  \frac{2}{r^2d_A^3}\left( (\Tr[\bar{O}\bar{O}^\dagger])^2 d_B + \Tr[\bar{O}\bar{O}^\dagger\bar{O}\bar{O}^\dagger]r\right).
   \ee  
By \eqref{eq:second-D}, \eqref{eq:fourth-D} and  Hölder's inequality we have
\bb
    \big(\mathbb{E}\Tr\big[|D|\big]\big)^2 &\ge \frac{\big(\mathbb{E}\Tr\big[|D|^2\big]\big)^3}{\mathbb{E}\Tr\big[|D|^4\big]}
    \\&\ge \frac{(\frac{2}{r} \frac{\Tr\left[  \bar{O}^\dagger  \bar{O} \right]}{d_A})^3}{4^4\frac{2}{r^2d_A^3}\left( (\Tr[\bar{O}\bar{O}^\dagger])^2 d_B + \Tr[\bar{O}\bar{O}^\dagger\bar{O}\bar{O}^\dagger]r\right) }  
     \\&\ge \frac{4}{4^4r} \cdot \frac{d_A\eps^2}{d_B + 2\frac{r}{d_A}}
     \\&\ge \frac{2}{4^4} \eps^2
\ee
where we use $\Tr[\bar{O} \bar{O}^\dagger]= d_A(\epsilon^2+O(\eps^3))=\Omega (d_A\eps^2)$.
    and  $\Tr[\bar{O}\bar{O}^\dagger\bar{O}\bar{O}^\dagger] = d_A (\eps^4 + \Theta(\eps^6))$.

Finally, 
\bb
   \left\|J_{{\Phi}_1} -J_{{\Phi}_1}\right\|_{1}&\ge  \| D(U_1,U_2)\|_1 -2\eps^2 \ge \Omega(\eps). 
\ee

\section{Weingarten Calculus} \label{sec:weingarten facts}
As we use a random channel constructed from sampling a $\Haar$-random unitary matrix in our lower bound proofs, we need some facts from Weingarten calculus in order to compute the corresponding expectation values with respect to the Haar measure.  
If $\pi\in S_n$ is a permutation of $[n]$, let $\operatorname{Wg}(\pi,d)$ denote the Weingarten function of dimension $d$. The following lemma is useful for our results. 
\begin{lemma}[{\cite{gu2013moments}}]\label{lem:Wg} Let $U$ be a $\Haar$-distributed unitary $(d\times d)$-matrix and let $\{A_i,B_i\}_{i=1}^n$ be a sequence of complex $(d\times d)$-matrices. We have the following formula for the expectation value:
\bb\label{eq:wg_lem}
&\ex{\Tr(UB_1U^\dagger A_1U\dots UB_nU^\dagger A_n)}
\\&\qquad =\sum_{\alpha,\beta \in S_n}\operatorname{Wg}(\beta\alpha^{-1},d)\Tr_{\beta^{-1}}(B_1,\dots,B_n)\Tr_{\alpha\gamma_n}(A_1,\dots,A_n),
\ee
where $\gamma_n=(12\dots n)$ and, writing $\sigma$ in terms of cycles $\{C_j\}$ as $\sigma=\prod_j C_j $,
\bb
\Tr_{\sigma}(M_1,\dots,M_n)\coloneqq\prod_j \Tr\prod_{i\in C_j} M_i.
\ee
\end{lemma}
We will also need some values of Weingarten function.
\begin{lemma}[\cite{collins2006integration}]\label{lem:wg2}
The function $\operatorname{Wg}(\pi,d)$ has the following values:
\begin{itemize}
    \item $\operatorname{Wg}((1),d)=\frac{1}{d}$,
    \item $\operatorname{Wg}((12),d)=\frac{-1}{d(d^2-1)}$,
    \item $\operatorname{Wg}((1)(2),d)=\frac{1}{d^2-1}$.
\end{itemize}
\end{lemma}

\end{document}